\documentclass[11pt]{amsart}
\reversemarginpar
\usepackage{amsmath}
\usepackage{amssymb}
\usepackage{graphicx}
\usepackage{graphicx}
\usepackage{dcolumn}
\usepackage{bm}
\usepackage{amsfonts}
\usepackage{latexsym}
\usepackage{pdfsync}
\usepackage{color}

\begin{document}
\newcommand{\commentout}[1]{}

\newcommand{\nwc}{\newcommand}
\newcommand{\bz}{{\mathbf z}}
\newcommand{\sqk}{\sqrt{\ks}}
\newcommand{\sqkone}{\sqrt{|\k\alpha|}}
\newcommand{\sqktwo}{\sqrt{|\k\beta|}}
\newcommand{\invsqkone}{|\k\alpha|^{-1/2}}
\newcommand{\invsqktwo}{|\k\beta|^{-1/2}}
\newcommand{\partz}{\frac{\partial}{\partial z}}
\newcommand{\grady}{\nabla_{\by}}
\newcommand{\gradp}{\nabla_{\bp}}
\newcommand{\gradx}{\nabla_{\bx}}
\newcommand{\invf}{\cF^{-1}_2}
\newcommand{\myphi}{\tilde\Theta_{(\eta,\rho)}}
\newcommand{\minrg}{|\min{(\rho,\phi^{-1})}|}
\newcommand{\al}{\alpha}
\newcommand{\xvec}{\vec{\mathbf x}}
\newcommand{\kvec}{{\vec{\mathbf k}}}
\newcommand{\lt}{\left}
\newcommand{\ksq}{\sqrt{\ks}}
\newcommand{\rt}{\right}
\newcommand{\ga}{\phi}
\newcommand{\vas}{\varepsilon}
\newcommand{\lan}{\left\langle}
\newcommand{\ran}{\right\rangle}
\newcommand{\tvas}{{W_z^\vas}}
\newcommand{\psiep}{{W_z^\vas}}
\newcommand{\wep}{{W^\vas}}
\newcommand{\weptil}{{\tilde{W}^\vas}}
\newcommand{\wepz}{{W_z^\vas}}
\newcommand{\weps}{{W_s^\ep}}
\newcommand{\wepsp}{{W_s^{\ep'}}}
\newcommand{\wepzp}{{W_z^{\vas'}}}
\newcommand{\wepztil}{{\tilde{W}_z^\vas}}
\newcommand{\vvas}{{\tilde{\ml L}_z^\vas}}
\newcommand{\veptil}{{\tilde{\ml L}_z^\vas}}
\newcommand{\vep}{{{ V}_z^\vas}}
\newcommand{\cvc}{{{\ml L}^{\ep*}_z}}
\newcommand{\cvcp}{{{\ml L}^{\ep*'}_z}}
\newcommand{\cvp}{{{\ml L}^{\ep*'}_z}}
\newcommand{\cvtil}{{\tilde{\ml L}^{\ep*}_z}}
\newcommand{\cvtilp}{{\tilde{\ml L}^{\ep*'}_z}}
\newcommand{\vtil}{{\tilde{V}^\ep_z}}
\newcommand{\ktil}{\tilde{K}}
\newcommand{\n}{\nabla}
\newcommand{\tkappa}{\tilde\kappa}
\newcommand{\Om}{{\Omega}}
\newcommand{\bx}{\mb x}
\nwc{\bv}{\mb v}
\newcommand{\br}{\mb r}
\nwc{\bH}{{\mb H}}
\newcommand{\bu}{\mathbf u}
\nwc{\bxp}{{{\mathbf x}}}
\nwc{\byp}{{{\mathbf y}}}
\newcommand{\bD}{\mathbf D}
\nwc{\bS}{\mathbf S}
\newcommand{\bA}{\mathbf \Phi}
\nwc{\cO}{\mathcal{O}}
\nwc{\co}{\mathcal{o}}
\nwc{\bG}{{\mathbf G}}
\nwc{\bF}{{\mathbf F}}
\nwc{\bd}{\hat {\mathbf d}}
\nwc{\bR}{{\mathbf R}}
\nwc{\bh}{\mathbf h}
\nwc{\bj}{{\mb j}}
\newcommand{\bB}{\mathbf B}
\newcommand{\bC}{\mathbf C}
\newcommand{\bp}{\mathbf p}
\newcommand{\bq}{\mathbf q}
\newcommand{\by}{\mathbf y}
\nwc{\bI}{\mathbf I}
\nwc{\bP}{\mathbf P}
\nwc{\bs}{\mathbf s}
\nwc{\bX}{\mathbf X}
\nwc{\bE}{\mathbf E}
\nwc{\bZ}{\mathbf Z}
\newcommand{\pdg}{\bp\cdot\nabla}
\newcommand{\pdgx}{\bp\cdot\nabla_\bx}
\newcommand{\one}{1\hspace{-4.4pt}1}
\newcommand{\corr}{r_{\eta,\rho}}
\newcommand{\rinf}{r_{\eta,\infty}}
\newcommand{\rzero}{r_{0,\rho}}
\newcommand{\rzeroinf}{r_{0,\infty}}
\nwc{\om}{\omega}
\nwc{\thetatil}{{\tilde\theta}}

\nwc{\nwt}{\newtheorem}
\nwc{\xp}{{x^{\perp}}}
\nwc{\yp}{{y^{\perp}}}
\nwt{remark}{Remark}
\nwt{corollary} {Corollary}
\nwt{definition}{Definition} 
\nwc{\ba}{{\mb a}}
\nwc{\bal}{\begin{align}}
\nwc{\be}{{\mathbf e}}
\nwc{\ben}{\begin{equation*}}
\nwc{\bea}{\begin{eqnarray}}
\nwc{\beq}{\begin{eqnarray}}
\nwc{\bean}{\begin{eqnarray*}}
\nwc{\beqn}{\begin{eqnarray*}}
\nwc{\beqast}{\begin{eqnarray*}}
\nwc{\bPsi}{{\mathbf \Psi}}
\nwc{\eal}{\end{align}}
\nwc{\ee}{\end{equation}}
\nwc{\een}{\end{equation*}}
\nwc{\eea}{\end{eqnarray}}
\nwc{\eeq}{\end{eqnarray}}
\nwc{\eean}{\end{eqnarray*}}
\nwc{\eeqn}{\end{eqnarray*}}
\nwc{\eeqast}{\end{eqnarray*}}

\nwc{\bQ}{\mathbf Q}
\nwc{\ep}{\varepsilon}
\nwc{\eps}{\varepsilon}
\nwc{\ept}{\ep }
\nwc{\vrho}{\varrho}
\nwc{\orho}{\bar\varrho}
\nwc{\ou}{\bar u}
\nwc{\vpsi}{\varpsi}
\nwc{\lamb}{\lambda}
\nwc{\Var}{{\rm Var}}

\nwt{proposition}{Proposition}
\nwt{theorem}{Theorem}
\nwt{lemma}{Lemma}
\nwc{\nn}{\nonumber}
\nwc{\mf}{\mathbf}
\nwc{\mb}{\mathbf}
\nwc{\ml}{\mathcal}

\nwc{\IA}{\mathbb{A}} 
\nwc{\bi}{\mathbf i}
\nwc{\bo}{\mathbf o}
\nwc{\IB}{\mathbb{B}}
\nwc{\IC}{\mathbb{C}} 
\nwc{\ID}{\mathbb{D}} 
\nwc{\IM}{\mathbb{M}} 
\nwc{\IP}{\mathbb{P}} 
\nwc{\II}{\mathbb{I}} 
\nwc{\IE}{\mathbb{E}} 
\nwc{\IF}{\mathbb{F}} 
\nwc{\IG}{\mathbb{G}} 
\nwc{\IN}{\mathbb{N}} 
\nwc{\IQ}{\mathbb{Q}} 
\nwc{\IR}{\mathbb{R}} 
\nwc{\IT}{\mathbb{T}} 
\nwc{\IZ}{\mathbb{Z}} 
\nwc{\pdfi}{{f^{\rm i}}}
\nwc{\pdfs}{{f^{\rm s}}}
\nwc{\pdfii}{{f_1^{\rm i}}}
\nwc{\pdfsi}{{f_1^{\rm s}}}
\nwc{\chis}{{\chi^{\rm s}}}
\nwc{\chii}{{\chi^{\rm i}}}
\nwc{\cE}{{\ml E}}
\nwc{\cP}{{\ml P}}
\nwc{\cQ}{{\ml Q}}
\nwc{\cL}{{\ml L}}
\nwc{\cX}{{\ml X}}
\nwc{\cY}{{\ml Y}}
\nwc{\cW}{{\ml W}}
\nwc{\cZ}{{\ml Z}}
\nwc{\cR}{{\ml R}}
\nwc{\cV}{{\ml V}}
\nwc{\cT}{{\ml T}}
\nwc{\crV}{{\ml L}_{(\delta,\rho)}}
\nwc{\cC}{{\ml C}}
\nwc{\cA}{{\ml A}}
\nwc{\cK}{{\ml K}}
\nwc{\cB}{{\ml B}}
\nwc{\cD}{{\ml D}}
\nwc{\cF}{{\ml F}}
\nwc{\cS}{{\ml S}}
\nwc{\cI}{{\ml I}}
\nwc{\cM}{{\ml M}}
\nwc{\cG}{{\ml G}}
\nwc{\cH}{{\ml H}}
\nwc{\bk}{{\mb k}}
\nwc{\cbz}{\overline{\cB}_z}
\nwc{\supp}{{\hbox{supp}(\theta)}}
\nwc{\fR}{\mathfrak{R}}
\nwc{\bY}{\mathbf Y}
\newcommand{\mbr}{\mb r}
\nwc{\pft}{\cF^{-1}_2}
\nwc{\bU}{{\mb U}}
\nwc{\bPhi}{{\mathbf \Phi}}

\title{
The {\em MUSIC} Algorithm for Sparse Objects:  A Compressed Sensing Analysis}
\author{Albert C.  Fannjiang}
\thanks{The research is partially supported by
the NSF grant DMS - 0908535}. 
\email{
fannjiang@math.ucdavis.edu}

       \address{
   Department of Mathematics,
    University of California, Davis, CA 95616-8633}
   
       \begin{abstract}
 The MUSIC algorithm, and  its extension for imaging sparse {\em extended}  objects, with noisy data is analyzed
by compressed sensing  (CS) techniques. 
A thresholding rule is developed to augment the
standard MUSIC algorithm. 
The notion of restricted isometry property (RIP) and
an upper bound on the restricted isometry constant  (RIC)
are employed to establish sufficient conditions for the  exact localization  by MUSIC with or without noise.

In the noiseless case, the sufficient condition gives an upper bound on the numbers of random sampling and incident directions 
necessary for exact localization. In the noisy case,
the sufficient condition assumes additionally an upper
bound for the noise-to-object ratio
in terms of the RIC and the dynamic range of objects. 
This bound points to the superresolution capability
of the MUSIC algorithm. 
Rigorous comparison of performance between MUSIC and
 the CS minimization principle, Basis Pursuit Denoising (BPDN), is given.

In general,   the MUSIC algorithm guarantees  to recover,  with high probability, $s$ scatterers with  $n=\cO(s^2)$ random sampling and incident directions and sufficiently high frequency. 
 
For the  favorable imaging geometry
 where the scatterers are distributed on a transverse plane 
MUSIC guarantees to recover, with high probability,  $s$ scatterers with a median frequency and $n=\cO(s)$ random sampling/incident directions. 
 
Moreover,  for the problems of 
spectral estimation and source localizations 
both BPDN and 
MUSIC guarantee, with high probability, to identify exactly the frequencies
of random signals with the number  $n=\cO(s)$ of sampling  times.
However,   in the absence of
  abundant realizations of signals, BPDN is the preferred
  method for spectral estimation. Indeed, BPDN can
  identify the frequencies  approximately with
  just {\em one}  realization of signals with the recovery error at worst linearly proportional  to the noise level. 

 Numerical results confirm that  BPDN outperforms
 MUSIC in the well-resolved case while the opposite 
 is true for the under-resolved case, giving
 abundant evidences for 
 the superresolution capability of the MUSIC algorithm.
 
 Another advantage of MUSIC over BPDN is the former's flexibility with grid spacing
 and guarantee of {\em approximate} localization
of sufficiently separated objects in an arbitrarily refined 
grid. The localization error is bounded from above by
$\cO(\lambda s)$ for  general configurations and
by $\cO(\lambda)$ for objects distributed in a transverse
plane. 
 

       \end{abstract}
       
       \maketitle
       
     \section{Introduction}
     
     The MUSIC (standing for {\bf MU}ltiple-{\bf Si}gnal-{\bf C}lassification) algorithm is a well-known method in signal processing for estimating
the individual frequencies of multiple time-harmonic signals \cite{Che, The}. Mathematically, 
MUSIC is essentially a method of characterizing the range of the covariance matrix of the signals  (see Section \ref{sec:spec} for details).

MUSIC was originally developed  to estimate the
direction of arrival for 
source localization \cite{Sch}. 
Later,  
the MUSIC algorithm is extended to imaging of  point scatterers
\cite{Dev}.  
A proof of a sufficient condition for
 the exact recovery of the object support in the noiseless case
is given in  \cite{ KG} (see also
\cite{Kir}) which is reproduced in Proposition \ref{prop1}
below.  The performance guarantee  
 is general but qualitative in nature. Neither does it 
 take into account the effect of noise which
 is important for assessing the superresolution effect.

 The main 
purpose of this paper is to give a quantitative performance
evaluation for the MUSIC algorithm in terms of 
 how many data are needed and how they may  be 
collected in order to exactly recover the locations of
given (large) number of objects, be they sources, scatterers
or frequencies as well as how much noise the MUSIC algorithm can tolerate.
Our approach is based on recent advances in
compressed  sensing  theory (\cite{BDE, Can, Rau} and references therein)  and
applications to imaging (\cite{cis-simo, cis-siso, cs-par} and
references therein). 

A main result for localizing scatterers obtained in the present paper has the following flavor (more details later): Let $\xi_{\rm max}$ and 
 $\xi_{\rm min}$ be, respectively, 
the strengths of the strongest and weakest (nonzero) scatterers, $\delta^\pm_s$
the (generalized) restricted isometry constants (RIC) of order $s$
and $\ep$ the level of noise in the data. 
If
the noise-to-scatterer ratio ({\rm NSR}) obeys the upper bound
\beq
 \label{251}
 {\ep\over \xi_{\rm min}}  < 
 \sqrt{{(1+\delta^+_s)^2}{\xi^2_{\rm max}\over \xi^2_{\rm min}} +(1-\delta^-_s)^2\Delta}
 -{(1+\delta^+_s)}{\xi_{\rm max}\over \xi_{\rm min}}
 \eeq
 where 
 \beqn
 \Delta&=&{1\over 2}-{1\over 2}{1\over \sqrt{\sqrt{2}\Gamma_\cS+1}}
 \eeqn
 and $\Gamma_\cS$ (defined in (\ref{90})) is a measure
 of the independence of the column vectors {\em outside
 the object support}  from the range
 of the data matrix
  then the MUSIC imaging function $J^\ep$  with 
 the thresholding rule
 \beq
 \label{a251}
 \lt\{ \br\in \cK: J^\ep(\br) \geq  2 \lt(1-{\delta^-_{s+1}(1+\delta^+_s)\over 2+\delta_s^+-\delta^-_{s+1}}\rt)^{-2}\rt\}
  \eeq
  recovers  exactly
 the locations of 
the $s$ scatterers (cf. Theorem \ref{cor2.1}, Section 3). 
Compressed sensing theory comes into play in addressing 
the dependence of RIC on the frequency, the number and distribution of random sampling directions (or sensors), the number of
scatterers
and the  inter-scatterer distances.

In the super-resolution regime, the
$\delta_s^-$ tends to $1$ and $\Gamma_\cS$ tends to zero,
rendering 
the right hand side of (\ref{251}) approximately 
\beq
\label{super'}
{(1-\delta_s^-)^2\Delta\over 
2(1+\delta_s^+)\xi_{\rm max}/\xi_{\rm min}}
\eeq
where $\xi_{\rm max}/\xi_{\rm min}$ is the dynamic range
of the scatterers. 
For a NSR  smaller than (\ref{super'})
the $s$ scatterers can still be perfectly localized by
the MUSIC algorithm with the thresholding rule (\ref{a251})
where the threshold is approximately
\[
{2\over (1-\delta_{s+1}^-)^2} \lt({1+\delta_s^+\over 2+\delta_s^+}\rt)^2\gg 1. 
\]
Previous observation \cite{PT} and our numerical results (Section \ref{sec:num}) lend support to this superresolution effect of
the MUSIC algorithm.

First let us review the inverse scattering problem
and the MUSIC imaging method. 
\subsection{Inverse scattering}
     
     \commentout{
  \begin{figure}[t]
\begin{center}
\includegraphics[width=0.5\textwidth]{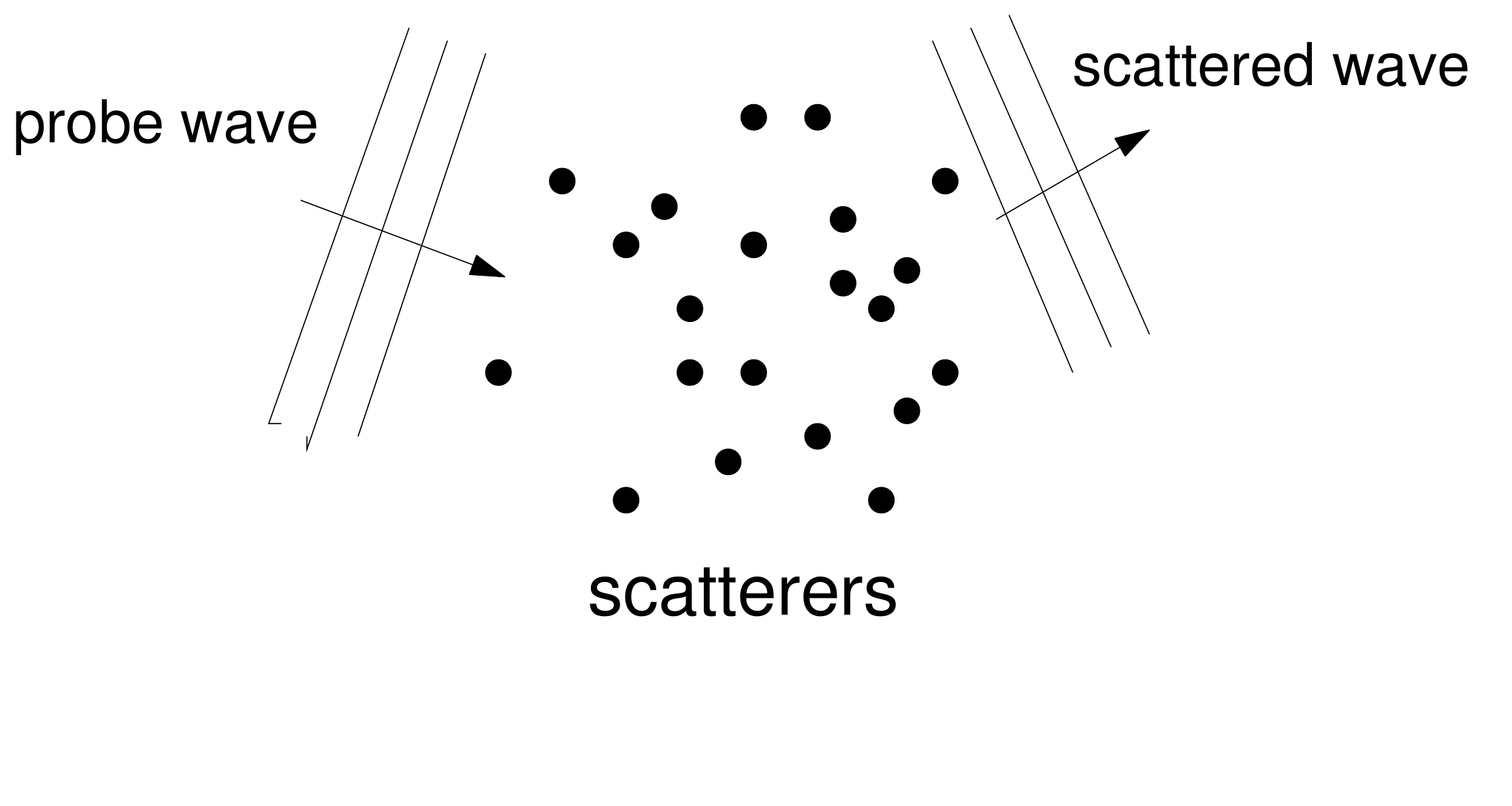}
\end{center}
\caption{Far-field  imaging of discrete  scatterers}
\label{fig-dt}
\end{figure}
}
 
 Consider the scattering of the incident plane wave
 \beq
 u^{\rm i}(\br)=e^{i\om \br\cdot\bd}
 \eeq
 by the  variable refractive index $n^2(\br)=1+\xi(\br)$
 where $\bd$ is the incident direction. 
The scattered field  satisfies 
the Lippmann-Schwinger equation 
\beq
\label{exact'}
u^{\rm s}(\br)&=&\om^2\int_{\IR^d} \xi(\br') 
\lt(u^{\rm i}(\br')+u^{\rm s}(\br')\rt) G(\br, \br')d\br', \quad d=2,3
\eeq
where $G(\br,\br')$ is the Green function
of the operator $-(\Delta+\om^2)$ \cite{KG}. 
We assume that the wave speed is unity and hence 
the frequency equals the wavenumber $\om$.

The scattered field has
the far-field asymptotic
\beq
u^{\rm s}(\br)={e^{i\om |\br|}\over |\br|^{(d-1)/2}}\lt(A
(\hat\br,\bd)+\cO(|\br|^{-1})\rt),\quad \hat\br=\br/|\br|,
\eeq
where the scattering amplitude $A$ is determined by the formula  
\beq
\label{sa}
A(\hat\br,\bd)&=&{\om^2\over 4\pi}
\int_{\IR^d} d\br' \xi(\br') u(\br') e^{-i\om \br'\cdot\hat\br}. 
\eeq
In the Born regime, the total field $u$ on
the right hand side of (\ref{sa}) can be  replaced by the
incident field $u^{\rm i}$. 

The main objective of inverse scattering then is to reconstruct
the medium  inhomogeneities  $\xi$ from the   knowledge
about the scattering amplitude $A(\hat\br,\bd)$.

  \begin{figure}[t]
\begin{center}
\includegraphics[width=0.5\textwidth]{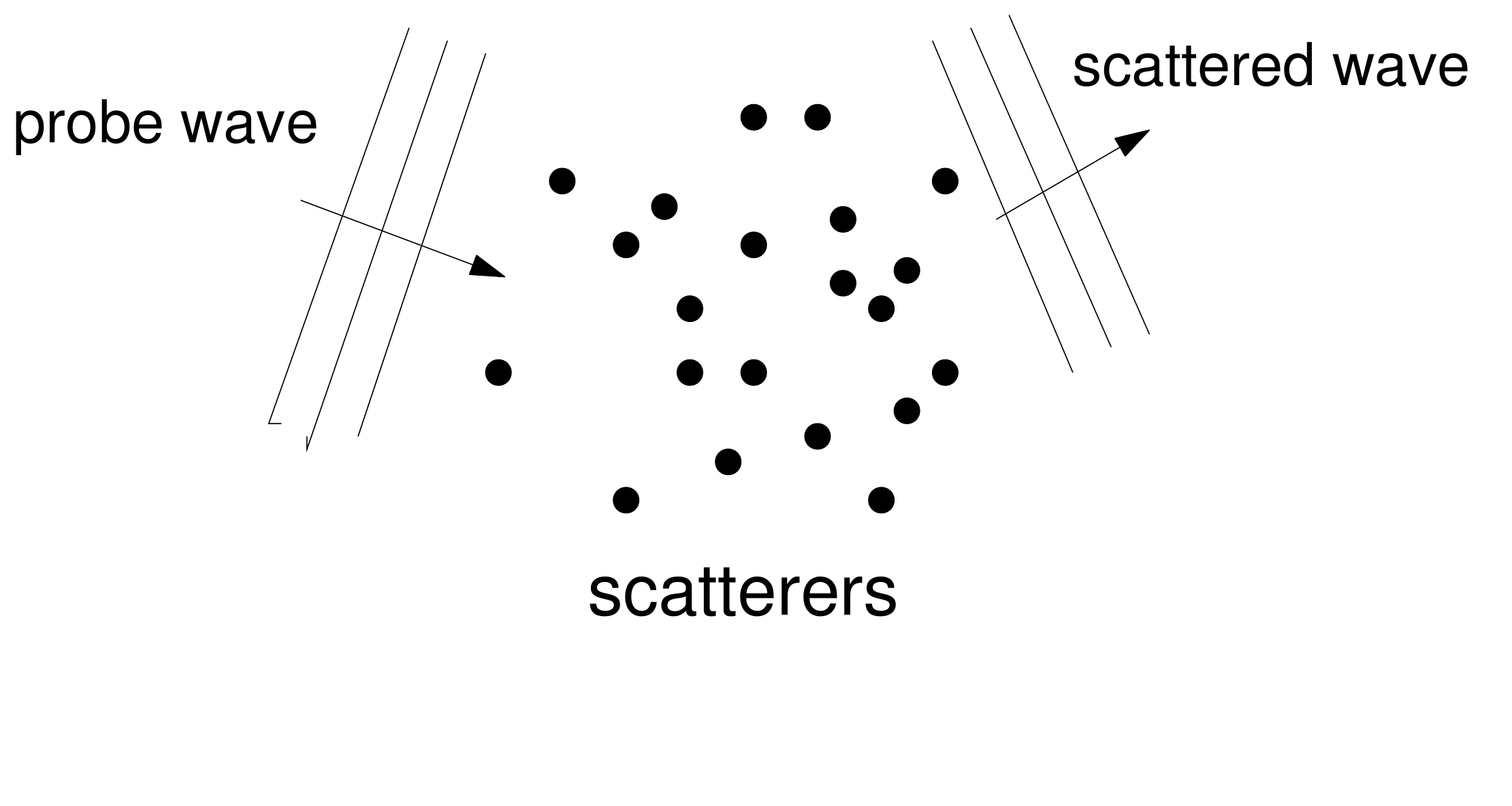}
\end{center}
\caption{Far-field  imaging geometry}
\end{figure}

Next we recall the MUSIC algorithm as applied to localization of
 point scatterers.
 \subsection{MUSIC for point scatterers}
\label{sec:mf}
\commentout{
We consider the medium with point scatterers located in 
 a square lattice 
\[
\cL=\lt\{\br_i=(x_i,z_i): i=1,...,m\rt\}
\]
of spacing  $\ell$. The total number $m$ of grid points
in $\cL$  is  a perfect square. 
Without loss of generality, assume $x_j=j_1\ell, z_j=j_2\ell $
where $j=(j_1-1)\sqrt{m}+j_2$ and $j_1, j_2=1,...,\sqrt{m}$.
Let  $\xi_{j}, j=1,...,m$ be the strength of
the scatterers.  Let  $
\cS=\lt\{\br_{i_j}=(x_{i_j}, z_{i_j}): j=1,...,s\rt\}$
be the
locations of the scatterers. Hence $\xi_j=0, \forall \br_j\not \in \cS$.  
}

Let  $
\cS=\lt\{\br_{j}: j=1,...,s\rt\}$
be the
locations of the scatterers.
Let  $\xi_{j}\neq 0, j=1,...,s$ be the strength of
the scatterers.  We will make the Born approximation  first
and discuss how to lift this restriction at the end of
the section (Remark \ref{rmk200}).
 For the discrete medium the scattering amplitude
becomes the finite sum 
\beq
A(\hat\br,\bd)&=&{\om^2\over 4\pi}
\sum_{j=1}^s \xi_{j} u^{\rm i}(\br_{j})
 e^{-i\om \br_j\cdot\hat\br_j}
\eeq
under the Born approximation.

Let $\bd_l, l=1,...,m$ and $ \hat\bs_k, k=1,...,n$ be, respectively, the incident and sampling directions.
For each incident field $\bd_l, l=1,...,m,$ the scattering amplitude
is measured in all $n$ directions $\hat\bs_k, k=1,...,n$. The whole measurement
data consist of the scattering amplitudes for all
pairs of $(\bd_l,\hat\bs_k)$. 

Define  the data matrix $\bY=(Y_{k,l})\in \IC^{n\times m}$ 
as
\beq
\label{4'}
Y_{k,l}\sim A(\hat\bs_k,\bd_l),\quad k=1,...,n,\quad
l=1,...,m
\eeq
where we keep open  the option of normalizing $\bY$ in order
to simplify the set-up. 
  The data matrix  is related to the object matrix
   \[
  \bX=\hbox{\rm diag} (\xi_j)\in \IC^{s\times s},\quad j=1,...s
   \]
     by the  measurement matrices  $\bA $ and $\bPsi$ as 
  \beq
  \label{31}
  \bY=\bA \bX \bPsi^*  
  \eeq 
  where  $\bA$ and $\bPsi$ are, respectively, 
     \beq
   \label{entry}
\Phi_{k, j}&=&{1\over \sqrt{n}}e^{-i\om\hat\bs_k\cdot \br_j}\in \IC^{n\times s}\\
\Psi_{l,j}&=&{1\over \sqrt{n}}e^{-i\om\bd_l\cdot \br_j}\in \IC^{m\times s}\label{entry2}
\eeq
after proper normalization. Both (\ref{entry})  and (\ref{entry2})  are normalized to
have columns of unit 2-norm. We extend the formulation (\ref{31})-(\ref{entry2}) to the case of sparse {\em extended} objects
in Appendix \ref{sec:ext-loc}. 

Note that both $\bA$ and $\bPsi$ are unknown and
(\ref{31}) can be inverted only after the locations of
scatterers are determined. This is what the MUSIC algorithm
is designed to accomplish. 

The standard version of MUSIC algorithm deals with the
case of $n=m$ and $\hat\bs_k=\bd_k, k=1,...,n$ as stated in
the following result.

\begin{proposition} \label{prop1} \cite{Kir, KG} Let $\{\hat\bs_k=\bd_k, k\in\IN\}$ be a countable
set of directions such that any analytic function on the unit sphere that vanishes in $\hat\bs_k, \forall k\in \IN$ vanishes
identically. 
Let $\cK\subset \IR^3$ be a compact subset containing $\cS$. Then there exists $n_0$ such that for any
$n\geq n_0$ the following characterization holds for
every $\br\in \cK$: 
\beq
\label{61}
\br\in \cS \,\,\hbox{ if and only if }\,\,
 \phi_\br\equiv {1\over \sqrt{n}} (e^{-i\om \hat \bs_1\cdot\br}, 
 e^{-i\om \hat \bs_2\cdot\br},\cdots, e^{-i\om \hat \bs_n\cdot\br})^T
\in \hbox{\rm Ran}(\bA). 
\eeq
 Moreover,
the ranges of $\bA$ and $\bY$ coincide. 
\end{proposition}
\begin{remark}
As a consequence, 
$\br\in \cS $ if and only if $\cP\phi_\br=0$
where $\cP$ is the orthogonal projection 
onto the null space of $\bY^*$ (Fredholm alternative). And the locations
of the scatterers can be identified by
the singularities  of the imaging function
\beq
\label{mus}
J(\br)={1\over |\cP\phi_\br|^2}
\eeq
\cite{Che}. 

Moreover, once the locations are exactly recovered, then both
$\bA$ and $\bPsi$  are known explicitly and the strength $\xi_j, j=1,...,s$ 
of scatterers can be determined by inverting  
the linear equation (\ref{31}) which is an over-determined system.  
\end{remark}

\begin{remark} \label{rmk200}
The assumptions of Proposition \ref{prop1} can be relaxed:
instead of $\hat\bs_k=\bd_k,\forall k$, it suffices to
have  $\bPsi\in \IC^{m\times s}$ which has  rank $s$. 

In light of this observation, it is also straightforward
to extend the performance guarantee for the Born
scattering case to the multiple-scattering case. 
In the latter case, $\bPsi$ consists of entries
 which are the total field
evaluated at $\br_j$ for the incident direction $\bd_l$,  i.e.
\beq
\Psi_{l, j}=u^*(\br_j;\bd_l).\label{non-Born}
\eeq
What is really needed is that $\bPsi\in \IC^{m\times s}$ has
rank $s$ since then $\bZ=\bX\bPsi^*$ and $\bX$ share the
same support (see more on this in Section \ref{sec:sen}). Generically this is true for sufficiently large $m$  as we will
show below.

Define the incident and full field vectors
at the locations of the scatterers:
\beqn
 U^{\rm i}(\bd_l)&=&(u^{\rm i}(\br_{1};\bd_l),...,u^{\rm i}(\br_{s};\bd_l))^T\in \IC^s\\
U(\bd_l)&=&(u(\br_{1};\bd_l),...,u(\br_{s}:\bd_l))^T\in \IC^s.
\eeqn

Denote
\beq
\label{fl}
\bG=[(1-\delta_{ij})G(\br_{j},\br_{i})]\in \IC^{s\times s}. 
\eeq
The discrete version of the Lippmann-Schwinger equation (i.e. the
Foldy-Lax equation) can be written as
\beq
\label{fl2}
U(\bd_l)=U^{\rm i}(\bd_l)+\om^2 \bG\bX U(\bd_l),\quad l=1,...,m.
\eeq
The $\delta_{ij}$ terms in (\ref{fl}) represent the
singular self-energy terms of point scatterers and
should be  removed for self consistency.  

Denote $\bU^{\rm i}=[U^{\rm i}(\bd_1),...,U^{\rm i}(\bd_m)]\in  \IC^{s\times m}$
and $\bU=[U(\bd_1),...,U(\bd_m)]\in \IC^{s\times m}$. 
Suppose that $\om^{-2}$ is not an eigenvalue of
$\bG\bX$. Then we can invert  eq. (\ref{fl2}) to obtain
\[
\bU=(\bI-\om^2\bG\bX)^{-1} \bU^{\rm i}.
\]
Hence $\bU$ has rank $s$ if $\bU^{\rm i}$ does.
Indeed,  for sufficiently high frequency
$\om$ and $m$  randomly  selected
 incident directions with  sufficiently large
ratio $\sqrt{m}/s$,    $\bU^{\rm i}$ has rank $s$ with high probability  (Propositions
\ref{prop2} and \ref{thm1} below).

For some  special imaging geometry  it is possible
to reduce the number of incident and sampling directions
to $\cO(s)$ (Section \ref{sec:opt}). 

\end{remark}

\subsection{Outline}
Proposition \ref{prop1} says that if the number
of sampling directions is sufficiently large then  
 the locations of the $s$ scatterers can be identified  by
the $s$ singularities of $J$.  However, the condition is only qualitative
in the sense that an estimate for the threshold $n_0$ is
not given. It would be of obvious interest to know, e.g. 
how $n_0$ scales with $s$ when $s$ is large
and when the conventional wisdom ($n_0=s+1$) derived  from
counting dimensions is true.  Also, how much noise
can the MUSIC algorithm tolerate?

However,  
unless additional constraints  are imposed on the
measurement scheme (the frequency, the incident
and 
sampling directions etc), it is unlikely to make progress 
toward obtaining an useful estimate which
is the objective of the present study.  In \cite{Dev} a geometric constraint on
the configuration of  sensors and objects  has been pointed out
for exact recovery in the absence of noise. Moreover, it seems  possible that  a non-vanishing portion of $s$ randomly distributed  scatterers may not be exactly recovered
 in the presence of machine error
no matter how large $n$ is (Figure \ref{fig6}, middle panel, and
Figure \ref{fig7}, right panel).  
\commentout{Also, the folklore about MUSIC is  that it should
work well for ``widely separated''  objects. Can one  give a precise
formulation and analysis of this case?
}

Let us briefly sketch our approach and results:{\em 
We shall discretize  the problem by using a finite grid for
the computation domain $\cK$ and put the problem
in a probabilistic setting by using
random sampling directions. Moreover, we consider
noisy data and aim for a result for stable recovery by MUSIC. For the case of
 well-resolved grids, we show  by using the compressed sensing techniques 
that  for the NSR obeying (\ref{251}) and  with high probability,  
$n_0=\cO(s^2)$  for
general configuration of $s$ objects 
and $n_0=\cO(s)$ for objects  distributed on a transverse plane. For the case of under-resolved grids, we seek
sufficient conditions for approximate, instead of exact, 
localization of objects and 
we show that for sufficiently small   {\rm NSR} and with high probability,  
the localization error is $\cO(\lambda s)$ with 
$n=\cO(s^2)$ for a general object configuration and  the localization error is  $\cO(\lambda)$ with $n=\cO(s)$ for
objects distributed in a transverse plane.
}

Our  plan for the rest of the paper  is to first give a sensitivity analysis for MUSIC and derive the condition for exact recovery
with noisy data under which
 the MUSIC algorithm based on the perturbed data matrix
 can still recover exactly the object support (Section \ref{sec:sen}).
 Next,  we review
 the basic notion  of compressed sensing (CS) theory
 and show how it naturally lends itself to a proof of
 exact localization  by MUSIC (Section \ref{sec:cs}). 
 We show that with generic, random sampling and sufficiently
 high frequency  the MUSIC algorithm
 can, with high probability,  recover $s$ scatterers with  $n=\cO(s^2)$ sampling and incident directions (Corollary \ref{cor1}). Then we consider a favorable imaging geometry
 where the scatterers are distributed on a transverse plane 
 (Section \ref{sec:opt}).
 We show that with median frequency the MUSIC algorithm
 can recover, with high probability,  $s$ scatterers with $n=\cO(s)$ sampling and incident directions
(Corollary \ref{cor2}). Next
 we analyze the  performance guarantee 
 of the compressed sensing principle, Basis Pursuit Denoising (BPDN)  (Section \ref{sec:bp})  and show that
 in the generic situation BPDN with sufficiently high frequency can recover $s$ scatterers with  $n=\cO(s^2)$ sampling directions and just one incident wave (Remark \ref{rmk7}) while
 for the favorable geometry of planar objects  BPDN can recover
 $s$ scatterers with $n=\cO(s)$  sampling directions and 
 one incident wave (Remark \ref{rmk6}). 
 In Section \ref{sec:spec} we return to the original applications 
 of MUSIC and
 perform the compressed sensing analysis of the
 performance of MUSIC as applied to
spectral estimation and source localization. We show that
the MUSIC algorithm can, with high probability,  identify exactly the frequencies
of random signals with the number  $n=\cO(s)$ of sampling times (Corollary \ref{cor4} and Remark \ref{rmk9}). 
We discuss MUSIC in the setting with an arbitrarily fine
grid and give error bounds in Section \ref{sec:grid}.  Numerical tests are given
in Section \ref{sec:num} where the superresolution capability
of 
MUSIC and the noise sensitivity are studied. We conclude in
Section \ref{sec:con}. We give an extension of the MUSIC
algorithm to the case of extended objects in Appendix \ref{sec:ext-loc} and a proof of performance guarantee for
BPDN in Appendix \ref{app:B}.

\section{Sensitivity analysis}
\label{sec:sen}

For  quantitative performance analysis of the MUSIC algorithm,  we will work with
the discrete setting and assume  that $\cK$ is a discrete
set of $N$, typically  large, number of points, i.e. 
the computation grid. The discrete setting appears naturally in applying MUSIC to imaging of extended scatterers (see Appendix \ref{sec:ext-loc}). Moreover,  we 
consider the extension $\tilde \bA$ of $\bA$ 
which includes not only the columns $\phi_{\br_j}, j=1,...,s$  representing the
locations of the objects but also the columns representing
all the points in  $\cK$.  Hence
$\tilde\bA\in \IC^{n\times N}$ and as usual 
$\tilde\bA$ is normalized  so that the columns have
unit 2-norm. The ordering of the columns of
 $\tilde\bA$  is not important  for our purpose as long as they
 correspond to the points in $\cK$ in a well-defined manner.  $\tilde\bPsi\in \IC^{n\times N}$ is 
similarly defined. Also the extension $\tilde \bX \in \IC^{N\times N}$ of $X$ is defined by filling in zeros in all
the entries outside the object support. 

In terms of these notations, we can write
\beq
\label{15'}
\bY=\tilde\bA\tilde \bX\tilde\bPsi^*=\sum_{j=1}^N \tilde\Phi_j\otimes\tilde\Psi_j^*\xi_j.
\eeq

By a slight abuse of notation,
we shall use $\cS$ to denote the locations of objects  in the
physical domain as well as 
the corresponding  index set. Likewise $\cS^c$ denotes
the complement set of $\cS$ in the computation grid $\cK$
as well as the total index set $\{1,...,N\}$. In the same vein, 
 $\tilde\bA_\cS$ denotes  the column submatrix of $\bA$ 
restricted to the index set $\cS$. Hence $\tilde\bA_\cS=\bA$
and $\tilde\bPsi_\cS=\bPsi$. 

First, let us reformulate the 
condition  (\ref{61}) for 
exact recovery as follows.

Note that 
\[
\bA\bA^\dagger \phi_\br
\]
 is the orthogonal projection of $\phi_\br$ onto
 the range of $\bA$ where $\bA^\dagger$ is
 the pseudo-inverse of $\bA$. Hence
 \[
 \cP\phi_\br=(\bI-\bA\bA^\dagger )\phi_\br.
 \]
If $\phi_\br$
for $\br \in \cS^c$  is independent
of the columns of $\bA$, then
\[
\phi_\br^*\bA \bA^\dagger \phi_\br<\|\phi_\br\|^2
\]
and vice versa. Therefore (\ref{61}) is equivalent to 
\beq
\label{90}
\Gamma_\cS\equiv \min_{\br\in \cS^c} \|\phi_\br\|_2^{-1} \|\cP\phi_\br\|_2=\sqrt{1-\max_{\br\in \cS^c} \|\phi_\br\|^{-2}\phi_\br^*\bA \bA^\dagger \phi_\br} >0. 
\eeq
The number $\Gamma_\cS$  gives a measure of how ``independent"  $\phi_z$  is from the range of $\bA$
uniformly in $\br\in \cS^c$.  
 
 \commentout{
\section{Planar objects}

We assume that the object is distributed on the transverse
plane $z=0$ and 
 the object function  $\xi(\br)$  admits  
 the Fourier expansion 
 \beqn
 \label{0.1}
 \sigma(\br)=\sum_{\bk\in \IZ^2} \hat \xi_{\bk}
 e^{i2\pi \bk\cdot \br/L}. 
 \eeqn
The scattered field at $z=z_0$  is given by
 \beq
 \label{15'}
 u_{\rm s}(z_0, \br)=\int_{\IR^2} G(z_0, \br-\br') \xi(0,\br') u^{\rm i}(0, \br') d\br'. 
 \eeq
 
Consider the obliquely  incident  plane wave  $u^{\rm i}(z,\br)=e^{i\om(\alpha x+\beta y+\gamma z)}$ where $\gamma=\sqrt{1-\alpha^2-\beta^2}$ with
\[
(\alpha,\beta)={\lambda\over L}\bq,\quad |\bq| < {L\over\lambda}.
\]
  The Green function 
$G$ can be expressed by
  the Weyl integral representation
  \beq
  \label{somm}
 G(x,y,z)={i\over 4\pi}
  \int e^{i\om(|z|\gamma(\alpha,\beta)+x\alpha+y\beta)} {d\alpha d\beta \over \gamma(\alpha,\beta)}  \eeq
  where
$\gamma$  is related to $\alpha,\beta $ by
   \beq
  \label{16'}
  \gamma(\alpha,\beta)=\lt\{\begin{array}{ll}
  \sqrt{1-\alpha^2-\beta^2}, & |\alpha|^2+|\beta|^2<1\\
  i\sqrt{\alpha^2+\beta^2-1},& |\alpha|^2+|\beta|^2> 1. 
  \end{array}
  \rt.
  \eeq
The integrand in (\ref{somm}) with  
  real-valued $\gamma$
corresponds to the homogeneous wave
and that with imaginary-valued $\gamma$ corresponds
to the evanescent (inhomogeneous) wave which has
an exponential-decay factor $e^{-\om |z| \sqrt{\alpha^2+\beta^2-1}}$. 

Using (\ref{16'}) to calculate the integral  in (\ref{15'}) we obtain
\beqn
u_{\rm s}(0,x,y)={i\over 2\om}\sum_{\bk} {\hat\xi_\bk\over \gamma_{\bk+\bq}}  e^{i\om z_0\gamma_{\bk+\bq}} e^{i\om  (x\alpha_{\bk+\bq}+y\beta_{\bk +\bq} )} 
\eeqn
where
\beqn
(\alpha_{\bk+\bq},\beta_{\bk+\bq})= {(\bk+\bq) \lambda\over L}, & & \gamma_{\bk+\bq} =
\lt\{\begin{array}{ll}
  \sqrt{1-\alpha_{\bk+\bq}^2-\beta_{\bk+\bq}^2}, & |(\alpha_{\bk+\bq},\beta_{\bk+\bq})|<1\\
  i\sqrt{\alpha_{\bk+\bq}^2+\beta^2_{\bk+\bq}-1},& |(\alpha_{\bk+\bq},\beta_{\bk+\bq})|> 1
  \end{array}
  \rt.. 
  \eeqn

Let $(z_0, x_j,y_j), (x_j,y_j)=(\xi_j,\eta_j)L, j=1,...,n$ be the coordinates
of the $n$ sensors. 
}

Now we give a sensitivity analysis 
for MUSIC with respect to perturbation in  the
data matrix $\bY$  in terms of $\Gamma_\cS$ and
other parameters. We want to show what else 
is needed, in addition to (\ref{61}), to guarantee
exact recovery of the support of scatterers when
the data matrix is perturbed. 

 The general data matrix considered in this paper has
 the form $\bY^\ep=\bY+\bE$ where
 $\bY=\bA \bX\bPsi^*\in \IC^{n\times m}, m\geq s$, the number  of objects. 
 Set $\bZ=\bX \bPsi^* \in \IC^{s\times m}$ such that  $\bY=\bA \bZ$ where $\bZ$ is assumed to have rank $s$. 
 
We shall treat $\bZ$ as the new object matrix and
consider 
perturbed  data matrices of the form 
\beq
\label{20'}
\bY^\ep=\bA\bZ+\bE. 
\eeq
Note that the locations of objects represented by $\bZ$ are
identical to those represented by $\bX=\hbox{\rm diag}(\xi_j)$.

Set 
 \[
 \cY^\ep=\bY^\ep\bY^{\ep*}=\cY+\cE
 \]
 where
 \beq
 \cY&=&\bA\bZ\bZ^*\bA^* \in \IC^{n\times n}\label{200}\\
 \cE&=& \bE\bZ^*\bA^*+\bA\bZ\bE^*+\bE\bE^*\in \IC^{n\times n}\label{201}
 \eeq
 are both self-adjoint.  
 Note that the range of $\cY$ is the same as the range of
 $\bA$ and under the assumption of (\ref{61})
 equals to the span of $\{\phi_\br:\br \in \cS\}$.

Let  $\{\bv_j: j=1,...,s\}$ and $\{\bv_j: j=s+1,...,n\}$, respectively,  be the set of orthonormal  bases  for the range and null space  of  $\cY$.  Let $\bQ_1\in \IC^{n\times s}$ and $\bQ_2\in \IC^{ n\times (n-s)}$, respectively, be
 the matrices whose columns are exactly  $\{\bv_j: j=1,...,s\}$ and $\{\bv_j: j=s+1,...,n\}$. Let $\bQ=[\bQ_1,\bQ_2] \in \IC^{n\times n}$. 
 
Let $\sigma_1\geq \sigma_2\geq \cdots \geq \sigma_n$ be
the singular values of $\cY$. 
Denote  the {\em smallest nonzero}  singular value of $\cY$ by $\sigma_{\rm min}$ and set $\sigma_{\rm max}=\sigma_1$. If $\bY$ has rank $s$, then $\sigma_{\rm min}=\sigma_s$.  We partition $\bQ^* \cE\bQ$ as follows:
\beq
\bQ^* \cE\bQ=\lt[\begin{matrix}
\cE_{11}&\cE_{12}\\
\cE_{21}&\cE_{22}
\end{matrix}
\rt]
\eeq
where $\cE_{11}\in \IC^{s\times s}, \cE_{12}\in \IC^{s\times (n-s)}, \cE_{21}\in \IC^{(n-s)\times s}, \cE_{22}\in \IC^{(n-s)\times (n-s)}$.

The following is a slight recasting of a general result of  matrix perturbation theory  \cite{SS}.
\begin{proposition}(Theorem 2.7, Chap. V, \cite{SS})
\label{prop2'} 
If  
\beq
\label{205}
{\sqrt{\|\cE_{12}\|_2\|\cE_{21}\|_2}\over \sigma_{\rm min}-\|\cE_{11}\|_2-\|\cE_{22}\|_2}< {1\over 2}\eeq
 then there exist
 $\bF\in \IC^{(n-s)\times s}$ with
\beq
\label{83}
\|\bF\|_2&\leq& {2\|\cE_{21}\|_2\over \sigma_{\rm min}-\|\cE_{11}\|_2-\|\cE_{22}\|_2 }
\eeq
such that the columns of
\beq
\label{64}
\bQ^\ep_1&=&(\bQ_1+\bQ_2\bF)(\bI+\bF^*\bF)^{-1/2}\\
\bQ^\ep_2&=&(\bQ_2-\bQ_1\bF^*)(\bI+\bF\bF^*)^{-1/2}\label{64'}
\eeq
are, respectively,  orthornormal bases for  invariant subspaces
of $\cY^\ep$. 

The representation of $\cY^\ep$ with respect to $\bQ^\ep_1,\bQ^\ep_2$ is, respectively, 
\beq
\Sigma^\ep_1&=&(\bI+\bF^*\bF)^{1/2}
\lt[ \Sigma_1+\cE_{11}+\cE_{12} \bF\rt](\bI+\bF^*\bF)^{-1/2}\\
\Sigma_2^\ep&=&(\bI+\bF\bF^*)^{-1/2}
\lt[\cE_{22}-\bF\cE_{12}\rt](\bI+\bF\bF^*)^{1/2}
\eeq
where $\Sigma_1=\hbox{diag}(\sigma_1, \sigma_2,...,\sigma_{\rm min})$. 
\end{proposition}
\commentout{
\begin{proposition}
Let $\{\lambda_k\}$ be the complete set of eigenvalues of
$\cY$. Let
\[
D_k=\{\mu+\imath \nu:|\mu+\imath \nu-\lambda_k|\leq \|\cE\|_2 \,\,\&\,\,
|\nu|\leq {1\over 2}\|\cE-\cE^*\|_2\}.
\]
Then the eigenvalues of $\cY^\ep$ belong to the union
of $D_k$. 
\end{proposition}
}

\begin{corollary}\label{cor20} 
Let $\rho_*\in (1/5, 1/4)$ be the only real root of
the cubic polynomial $p(\rho)=1-8\rho+20\rho^2-20\rho^3$
and suppose
\beq
\label{206}
{\|\cE\|_2\over \sigma_{\rm min}}<\rho_*. 
\eeq
Then  $ \hbox{\rm Ran}(\bQ^\ep_1)$ is the singular subspace
associated with the $s$ largest singular values of $\bY^\ep$ and $ \hbox{\rm Ran}(\bQ^\ep_2)$ the singular subspace associated with the rest of the singular values. 
\end{corollary}
\begin{proof}It suffices to show that under (\ref{206})
the smallest singular value of $\Sigma_1$ is larger than the largest singular value of $\Sigma_2$. 

Since $\rho_*<1/4$, condition (\ref{206}) implies that
\[
{\|\cE\|_2 \over \sigma_{\rm min}-2\|\cE\|_2}<{1\over 2} 
\]
which in turn implies (\ref{205}). Note that
under this condition $\|\bF\|_2\leq 1$. 
By Proposition \ref{prop2'} we have
\beq
\|\Sigma^\ep_2\|_2&\leq &\|\bI+\bF\bF^*\|^{1/2}_2
\|\cE_{22}-\bF\cE_{12}\|_2\label{249'} \\
&\leq &\lt(1+{4\|\cE\|_2^2\over (\sigma_{\rm min}-2\|\cE\|_2)^2}\rt)^{1/2}
\lt(\|\cE\|_2+{2\|\cE\|_2^2\over \sigma_{\rm min}-2\|\cE\|_2}\rt)\nn
\eeq
On the other hand,
\beq
\label{249}
\min_{\|\be\|_2=1}\|\Sigma^\ep_1\be\|_2&=&\min_{ \|\be\|_2=1}\|(\Sigma_1+\cE_{11}+\cE_{12}\bF)\be\|
\|\bI+\bF^*\bF\|^{-1/2}_2\\
&\geq& \lt(\sigma_{\rm min}-\|\cE\|_2-{2\|\cE\|^2_2\over \sigma_{\rm min}-2\|\cE\|_2}\rt) \lt(1+{4\|\cE\|_2^2\over (\sigma_{\rm min}-2\|\cE\|_2)^2} \rt)^{-1/2}. \nn
\eeq

Let 
\[
\rho={\|\cE\|_2\over \sigma_{\rm min}}. 
\]
Imposing that the right hand side of (\ref{249}) is greater
than that of (\ref{249'}) leads to
the inequality
\beqn
p(\rho)&=&(1-2\rho)^3-2\rho(1-2\rho)^2-4\rho^3\\
&=&1-8\rho+20\rho^2-20\rho^3
>0
\eeqn
which holds for $\rho<\rho_* $, the only real  root of
$p(x)$. It is readily verified that $\rho_*\in (1/5, 1/4).$

 \end{proof}
 
 In view of Corollary \ref{cor20}, it is natural to call  $\hbox{\rm Ran}(\bQ^\ep_2)$  the {\em noise} subspace and $\hbox{\rm Ran}(\bQ_1^\ep)$ the  {\em signal} (or {\em object})
 subspace.

 Let $\cP^\ep$
be the orthogonal projection onto the noise subspace
and define the MUSIC imaging function for the noisy
data 
\[
J^\ep(\br)={1\over |\cP^\ep\phi_\br|^2}. 
\]

We are ready to  state the first main result of the paper.
\begin{theorem}
 \label{thma}Let $\cY^\ep=\cY+\cE$ where $\cY$ and $\cE$ are given by (\ref{200})-(\ref{201}). 
  
  Suppose $\phi_\br\in {\rm Ran}(\bA)$ if and only if $\br\in \cS$. Then the condition
  \beq
  \label{202}
  {\|\cE\|_2\over \sigma_{\rm min}} &<&\Delta={1\over 2}-{1\over 2}{1\over \sqrt{\sqrt{2}\Gamma_\cS+1}}
  \eeq
implies that  the $s$ highest peaks of $
 J^\ep(\br)$  coincide with the true locations of objects. Indeed,
 the object locations can be identified by the thresholding rule:
 \beq
 \label{241}
\lt\{\br\in \cK:  J^\ep(\br)\geq 2\Gamma_\cS^{-2} \rt\}.
\eeq
   \end{theorem}
   \begin{proof}
   Let $\{\bv_j+\delta\bv_j:j=s+1,...,n\}$ be the columns of $\bQ^\ep_2$.  Clearly, $\delta\bv_j$ is the $(j-s)$-th column  of
$\bQ^\ep_2-\bQ_2$. Now we have  
\beqn
\|\bQ^\ep_2-\bQ_2\|_2&\leq& 
\|\bQ_1\bF^*(\bI+\bF\bF^*)^{-1/2} \|_2
+\|\bQ_2(\bI+\bF\bF^*)^{-1/2}-\bQ_2\|_2\\
&\leq&\|\bQ_1\bF^*\|_2
+\|(\bI+\bF\bF^*)^{-1/2}-\bI\|_2
\eeqn
whose first term is bounded by $\|\bF\|_2$ and
whose second term is bounded by
\beqn
\|(\bI+\bF\bF^*)^{-1/2}-\bI\|_2
&=&\|(\bI+\bF\bF^*)^{-1/2} (\bI+(\bI+\bF\bF^*)^{1/2})^{-1}\bF\bF^*\|_2\\
&\leq& {1\over 2} \|\bF\|_2^2.
\eeqn
Hence 
\beq
\|\delta\bv_j\|_2&\leq&\|\bQ^\ep_2-\bQ_2\|_2\leq
\|\bF\|_2+{1\over 2} \|\bF\|_2^2\leq 
{2\rho(1-\rho)\over (1-2\rho)^2},  \quad j=s+1,...,n,\label{225}
\eeq
where $\rho=\|\cE\|_2/\sigma_{\rm min}$. 

By Corollary \ref{cor20} $\{\bv_j+\delta\bv_j:j=s+1,...,n\}$ is the set of singular vectors associated with the $n-s$ smallest singular values of $\bY$.
By definition, 
\beq
\|\cP^\ep \phi_\br\|_2^2&=&
\sum_{k=s+1}^n \lt|(\bv^*_k+\delta\bv^*_k)\phi_\br\rt|^2\nn\\
&=& \sum_{k=s+1}^n \lt|\bv^*_k\phi_\br\rt|^2+
2\sum_{k=s+1}^n \Re[\bv^*_k\phi_\br
\phi^*_\br \delta\bv_k]+ \|(\bQ_2^\ep-\bQ_2)^*\phi_\br\|_2^2
.\label{63}
\eeq
By assumption
the first  two terms on the right hand side of (\ref{63}) vanish
if and only if $\br\in \cS$. By (\ref{225})  the third term is bounded by
\beqn
 \|(\bQ_2^\ep-\bQ_2)^*\phi_\br\|_2^2 &\leq& 
{4\rho^2(1-\rho)^2\over (1-2\rho)^4}. 
 \eeqn
For $\br\not \in \cS$, 
\beqn
\sum_{k=s+1}^n \lt|(\bv^*_k+\delta\bv^*_k)\phi_\br\rt|^2
&= &\sum_{k=s+1}^n \lt|\bv^*_k\phi_\br\rt|^2 - \|(\bQ_2^\ep-\bQ_2)^*\phi_\br\|_2^2\\
&\geq&  \Gamma^2_\cS  - {4\rho^2(1-\rho)^2\over (1-2\rho)^4}
\eeqn
by (\ref{225}). 
Hence $
J^\ep(\br)$
has the following behavior
\beq
\label{81}
J^\ep(\br)&\geq&  {(1-2\rho)^4\over 4\rho^2(1-\rho)^2} ,\quad \br\in \cS\\
J^\ep(\br)&\leq&\lt( \Gamma^2_\cS  -{4\rho^2(1-\rho)^2\over (1-2\rho)^4}\rt)^{-1},\quad \br\in \cS^c. \label{82}
\eeq

Setting
\[
 {(1-2\rho)^4\over 4\rho^2(1-\rho)^2}>\lt( \Gamma^2_\cS  -{4\rho^2(1-\rho)^2\over (1-2\rho)^4}\rt)^{-1}
  \]
 we obtain the inequality
 \beq
 \label{ineq}
 \rho^2-\rho+{\Gamma_\cS\over 4\Gamma_\cS+2\sqrt{2}}>0
 \eeq
whose solution is (\ref{202}). Note that
$\Delta <1/5<\rho_*, \forall \Gamma_\cS\in [0,1]$.

   \end{proof}

\commentout{
\begin{remark}
one needs 
the notion of {\em separation} between
two sub-matrices  in place of  the notion of
the smallest nonzero singular value \cite{Ste}. We leave
this extension  to the interested  readers. 
\end{remark}
}

Condition (\ref{202}) would not be very useful unless $\|\cE\|_2$ 
 can be bounded from above and $\sigma_{\rm min}$, $\Gamma_\cS$ can be bounded from below
by other known or accessible quantities. This is what the
compressed sensing techniques enable us to do.

\section{Compressed sensing analysis}\label{sec:cs}

We now give
a quantitative evaluation of MUSIC based on compressed sensing theory.
 
A fundamental notion in compressed sensing is the restrictive isometry property (RIP) due to Cand\`es and Tao \cite{CT}.
Precisely, let  the sparsity $s$  of a vector  $Z\in \IC^N$ be the
number of nonzero components of $Z$ and define the restricted isometry constants $\delta^-_s\in [0,1],\delta^+_s\in [0,\infty)$
to be the {\em smallest} nonnegative  numbers such that the inequality
\beq
\label{rip}
 (1-\delta^-_s) \|Z\|_2^2\leq \|\tilde\bA Z\|_2^2\leq (1+\delta^+_s)
\|Z\|_2^2
\eeq
holds for all $Z\in \IC^N$ of sparsity at most $ s$. 

Roughly speaking this means that $\tilde\bA$  acts like a near isometry, up to a scaling, 
when restricted to $s$-sparse vectors. 
In particular, if $\delta^-_{s+1}< 1$ then any $s+1$
columns of $\tilde \bA$ are linearly independent which
implies the characterization
(\ref{61}).

More generally, let us  extend the notion of the restricted isometry constants to ones $\delta^\pm_\cS$
associated with a particular set $\cS$, namely  the
smallest nonnegative numbers  satisfying 
\beq
\label{rip'}
 (1-\delta^-_\cS) \|Z\|_2^2\leq \|\tilde\bA Z\|_2^2\leq (1+\delta^+_\cS)
\|Z\|_2^2
\eeq
for all $Z\in \IC^N$ supported on the set $\cS$.
This will become important later when we analyze the case
of an arbitrarily refined  grid (Section \ref{sec:grid}). 
Clearly, 
\beq\label{112}
\delta^\pm_s=\max_{|\cS|=s}\delta^\pm_\cS. 
\eeq
Then (\ref{61}) is equivalent to $\delta^-_{\cS'}<1$ for
all $\cS'$ which is  the union of $\cS$ and another point $\br\in \cS^c$.

First,  let us estimate the magnitude of the  error term $\cE$
in terms of $\bE$ as follows.

\begin{lemma}
\label{lem5}
Suppose (\ref{rip'}) holds for $\bA$. 
Then 
 \beq
 \label{211}
 \|\cE\|_2\leq \|\bE\|_2^2+2\zeta_{\rm max}\sqrt{1+\delta^+_{\cS}} \|\bE\|_2\leq \|\bE\|_2^2+2\zeta_{\rm max}\sqrt{1+\delta^+_s} \|\bE\|_2
 \eeq
 where $\zeta_{\rm max}=\|\bZ\|_2$ is the largest singular value of
 $\bZ$. 
 
 For the case of $\bZ=\bX\bPsi^*$ with  $\bPsi$ satisfying 
 (\ref{rip'}),  we have
  \beq
 \label{211'}
 \|\cE\|_2\leq \|\bE\|_2^2+2\xi_{\rm max}({1+\delta^+_{\cS}}) \|\bE\|_2\leq \|\bE\|_2^2+2\xi_{\rm max}({1+\delta^+_s}) \|\bE\|_2
  \eeq
  where $\xi_{\rm max}=\max_{i}|\xi_i|$. 
\end{lemma}
\begin{proof} First we have 
 \beqn
 \|\cE\|_2\leq \|\bE\|_2^2+2\|\bA\bZ \bE^*\|_2.
 \eeqn
The RIP (\ref{rip'}) then implies that
\[
\|\bA\bZ \bE^*\|_2^2\leq (1+\delta^+_{\cS}) \|\bZ\bE^*\|_2^2
\]
and thus
\beqn
 \|\cE\|_2&\leq &\|\bE\|_2^2+2\sqrt{1+\delta^+_{\cS}}  \|\bZ\bE^*\|_2\\
 &\leq& \|\bE\|_2^2+2\sqrt{1+\delta^+_{\cS}}  \|\bZ\|_2\|\bE^*\|_2. 
 \eeqn
 
 In the case of scattering objects  $\bZ=\bX\bPsi^*$ 
 \[
  \|\bZ\bE^*\|_2
=\|\bE \bPsi \bX^*\|_2\leq\|\bE\|_2\|\bPsi \bX^*\|_2  \leq \|\bE\|_2
\xi_{\rm max} \sqrt{1+\delta^+_\bS} 
\]
 provided that
$\bPsi$ also satisfies the RIP (\ref{rip'}). In this case, 
\beq
\label{210}
 \|\cE\|_2\leq \|\bE\|_2^2+2\xi_{\rm max} (1+\delta^+_{\cS})  \|\bE\|_2
 \eeq
 and hence (\ref{211'}).

\end{proof}

\begin{lemma} \label{prop4'}   The minimum nonzero
singular value $\sigma_{\rm min}$ of $\cY$ obeys
the lower bound
\beq
\label{226}
\sigma_{\rm min}\geq (1-\delta^-_\cS)\zeta^2_{\rm min}\geq
(1-\delta^-_s)\zeta^2_{\rm min}
\eeq
where
\beq
\label{209}
\zeta_{\rm min}=\min_{\be\in \IC^s}{ \|\bZ^*\be\|_2\over \|\be\|_2}. 
\eeq

\commentout{On the other hand, the maximum singular value $\sigma_{\rm max}$ of $\cY$ satisfies
\beq
\label{227}
\sigma_{\rm max}\leq (1+\delta_\cS) \zeta_{\rm max}^2\leq
(1+\delta_s) \zeta_{\rm max}^2.
\eeq}

 For the case of scattering objects $\bZ=\bX\bPsi^*$ with  $\bPsi$ satisfying 
 (\ref{rip'}),  we have
\beq
\label{209'}
\sigma_{\rm min}&\geq& (1-\delta^-_\cS)^2\xi^2_{\rm min}\geq
(1-\delta^-_s)^2\xi^2_{\rm min}
\eeq

\end{lemma}
\begin{proof}

Using  the max-min theorem  \cite{HJ}
 \beq
 \label{212}
 \sigma_s(\cY)=\max_{{\rm \tiny dim} \cH=s}
 \min_{\be\in  \cH} {\|\cY\be\|_2\over\|\be\|_2}
 \eeq
with $\cH=\hbox{\rm Ran}(\bA)$ we obtain
\beq
\sigma_{\rm min}(\cY)\geq 
\min_{\be\in \hbox{\rm Ran}(\bA)\atop \|\be\|_2=1}\|\bA\bZ \bZ^*\bPhi^*  \be\|_2. \label{69'}
\eeq

Let $\{\bu_j: j=1,...,s\}$ be the eigen-vectors of $\bA\bA^*$ associated
with the nonzero eigenvalues $\{\lambda^2_1\geq \lambda_2^2\geq ...\geq \lambda_s^2\}$ and form the orthonormal 
basis of $\cH$. Write $\be=\sum_{j=1}^s e_j \bu_j$ where
$\sum_{j}^s |e_j|^2=1$. 
We have
\[
\bA\bA^* \be=\sum_{j=1}^s \lambda_j^2e_j\bu_j
\]
and thus
\beq
\label{70}
\|\bA^* \be\|_2^2=\sum_{j=1}^s \lambda_j^2 |e_j|^2\geq
\lambda_{s}^2.
\eeq

It follows from (\ref{69'}), (\ref{rip'}) and (\ref{70})  that
\beq
\label{71}
\sigma_{\rm min}&\geq \sqrt{1-\delta^-_\cS} \|\bZ\bZ^*\bA^* \be\|_2
&\geq \sqrt{1-\delta^-_\cS} \zeta^2_{\rm min} \lambda_s
\eeq
by (\ref{209}).  
\commentout{In particular, for $\bZ=\bX\bPsi^*$ we have
\beq
\label{71'}
\sigma_{\rm min}&\geq \sqrt{1-\delta_\cS} \|\bZ\bZ^*\bA^* \be\|_2
&\geq \sqrt{1-\delta_\cS} \zeta^2_{\rm min} \lambda_s
\eeq
}
On the other hand, $\lambda_s^2 $ is exactly the smallest
eigenvalue of $\bA^*\bA\in \IC^{s\times s}$ and hence
by (\ref{rip'}) is bounded from below by $1-\delta^-_\cS$.   
Using this observation in (\ref{71}) we
obtain (\ref{226}). 

\commentout{
For $\sigma_{\rm max}$ we have the following calculation:
\beqn
\sigma_{\rm max}&=&\sup_{\|\be\|_2=1} \|\bA\bZ \be\|_2^2\\
&\leq& (1+\delta_\cS)\sup_{\|\be\|_2=1} \|\bZ\be\|_2^2\\
&\leq& (1+\delta_\cS) \zeta_{\rm max}^2
\eeqn
implying  (\ref{227}). }

In the case of scattering objects we can bound $\zeta^2_{\rm min}$ as 
\beqn
\zeta^2_{\rm min}&\geq& (1-\delta^-_\cS) \xi^2_{\rm min} \\
\zeta^2_{\rm max}&\leq& (1+\delta^+_\cS) \xi^2_{\rm max} 
\eeqn
 by using (\ref{rip'})  with $\bPsi$ and hence the result (\ref{209'}). 

\commentout{

Restricted to $\cH$, $\bA^*:\cH\to\IC^s$ is invertible. 
Let us write $Z=\bA^{\dagger *} Z'$ for some $Z'\in \IC^s$. 
In other words, $Z'=\bA^* Z$.  Recall the 
RIP (\ref{rip}) in a different form 
\[
(1-\delta_s)\|\bA^\dagger Z\|_2^2\leq \|Z\|^2_2
\leq (1+\delta_s)\|\bA^\dagger Z\|_2^2. 
\]
}

\commentout{
Choose $Z\in \hbox{\rm Ran}(\bA)$ such that
\[
\tilde\bA^* Z=\hbox{\rm sgn}( {\bX})
\]
where 
\[
\hbox{\rm sgn}({\bX})=\lt\{\begin{matrix}
{\xi_j\over |\xi_j|},& j=1,..,s\\
0,& j=s+1,...,N
\end{matrix}\rt.
\]

Indeed, we may set
\[
Z=\bA^{\dagger *} \hbox{\rm sgn}({\bX})
\]
which is well-defined as $\tilde\bA^*$ has rank $s$.
Substituting this trial vector in (\ref{var}) we obtain
\[
\sigma_{\rm min}\geq \|\bA^{\dagger *} \hbox{\rm sgn}({\bX})\|_2^{-1}
\]
}
\end{proof}

\commentout{
An immediate consequence of (\ref{226}) and (\ref{227}) is
the following.

\begin{corollary}
\label{cor30}
We have the general bound
\[
\kappa(\cY)\leq  {1+\delta_\cS\over 1-\delta_\cS} \kappa^2(\bZ^*)\leq {1+\delta_s\over 1-\delta_s} \kappa^2(\bZ^*)
\]
and in the case of scattering objects $\bZ=\bX\bPsi^*$ with $\bPsi$ satisfying (\ref{rip'}) 
\[
\kappa(\cY)\leq  \lt({1+\delta_\cS\over 1-\delta_\cS}\rt)^2 {\xi_{\rm max}^2\over \xi_{\rm min}^2} \leq \lt({1+\delta_s\over 1-\delta_s} \rt)^2{\xi_{\rm max}^2\over \xi_{\rm min}^2}. 
\]

\end{corollary}
}

Next we derive   a lower bound for
$\Gamma_\cS$ in terms of RIC. 
\begin{lemma} Fix $S, |S|=s$. 
Then the lower bound is valid
\beq
\label{88}
 \Gamma_\cS\geq 1-\max_{\br\in\cS^c\atop
 \cS'=\cS\cup \{\br\}}{\delta^-_{\cS'}(1+\delta^+_\cS)\over 2+\delta_\cS^+-\delta^-_{\cS'}}\geq  1-{\delta^-_{s+1}(1+\delta^+_s)\over 2+\delta^+_s-\delta^-_{s+1}}. 
\eeq
\label{prop3'}
\end{lemma}
\commentout{
\begin{remark}
Alternatively, by exactly the same argument we have
the lower bound
\beq
\label{88'}
 \Gamma_\cS\geq 1-\max_{\br\in\cS^c\atop
 \cS'=\cS\cup \{\br\}}{\delta_{\cS'}(1+\delta_\cS)\over 2+\delta_\cS-\delta_{\cS'}}. 
\eeq
\label{prop3''}
\end{remark}
}
\begin{proof}
Without loss of generality, suppose
$\cS=\{1,2,...,s\}$ and consider $\phi_\br=\tilde\Phi_{s+1}$. 
Let $\cS'=\cS\cup \{s+1\}$. 
Our subsequent analysis is independent of these choices modulo
inconsequential notational  change. 

Denote $\cP\phi_\br=\phi'$ and write the orthogonal decomposition 
\beq
\label{67}
\phi_\br=\phi'+\sum_{j=1}^s c_j \tilde\Phi_j. 
\eeq
Hence we can express $\phi'$ as
\[
\phi'=\tilde\Phi_{s+1}-\sum_{j=1}^s c_j \tilde\Phi_j=
\tilde\bA Z,\quad Z=(-c_1,-c_2,...,-c_s,1,0,...,0)^T\in \IC^N.
\]
Using (\ref{rip'}) for sparsity  $\cS'$  we obtain a lower bound for $\|\phi'\|_2$:
\beq
\label{66}
(1-\delta^-_{\cS'})(1+\sum_{j=1}^s |c_j|^2)
\leq\|\phi'\|_2^2.
\eeq

On the other hand, 
we have by the Pythagorean theorem that 
\beq
\label{68}
\|\phi_\br\|_2^2=\|\phi'\|_2^2+\|\sum_{j=1}^s c_j \tilde\Phi_j\|_2^2.
\eeq
Applying  (\ref{rip'}) for sparsity $\cS$ to the second term on the right hand side of (\ref{68}) we obtain 
\beqn
\|\phi_\br\|^2_2-\|\phi'\|_2^2=\|\sum_{j=1}^s c_j \tilde\Phi_j\|_2^2 \leq (1+\delta^+_\cS) \sum_{j=1}^s |c_j|^2
\eeqn
and hence
\beq
\label{69}
 \sum_{j=1}^s |c_j|^2\geq  (1+\delta^+_\cS)^{-1}\lt( \|\phi_\br\|^2_2-\|\phi'\|_2^2\rt). 
\eeq

Combining (\ref{69}) and (\ref{66})  we obtain
\[
\|\phi'\|_2^2 \geq (1-\delta^-_{\cS'}) \lt(1+{1-\|\phi'\|_2^2\over 1+\delta^+_\cS}\rt)
\]
which can be solved to yield 
\beq
\label{89}
\|\phi'\|_2^2\geq 1-{\delta^-_{\cS'}(1+\delta^+_\cS)\over 2+\delta^+_\cS-\delta^-_{\cS'}}.
\eeq
Minimizing (\ref{89}) over $\br\in \cS^c$ we obtain 
the first inequality in (\ref{88}).

The second inequality (\ref{88}) follows from (\ref{112}) and 
the observation (by differentiation) that the quantity
\[
{\delta^-_{\cS'}(1+\delta^+_\cS)\over 2+\delta^+_\cS-\delta^-_{\cS'}}
 \]
 is an increasing function of $\delta^-_{\cS'}\in [0,1]$ and $\delta^+_\cS\in [0,\infty)$ separately.

\end{proof}

Combining the preceding results  we have the following stability criterion  for exact recovery by  MUSIC.
 \begin{theorem}\label{cor2.1}  \label{thmb}
   Suppose $\delta^-_{s+1}<1$ (implying (\ref{61}))
   and $\|\bE\|_2=\ep$.

If the noise-to-object ratio (NOR) satisfies 
 \beq
 \label{203}\label{87}\label{100}
 {\ep\over \zeta_{\rm min}}
 <  
 \sqrt{{(1+\delta^+_s)}{\zeta_{\rm max}^2\over \zeta_{\rm min}^2}+{(1-\delta^-_s)\Delta}}
 -{\zeta_{\rm max}\over \zeta_{\rm min}}\sqrt{1+\delta^+_s}
 \eeq
 where $\Delta$ is given by (\ref{202})
 then  the object support $\cS$  can be identified by
 the thresholding rule
 \beq
 \lt\{ \br\in \cK: J^\ep(\br) \geq  2 \lt(1-{\delta^-_{s+1}(1+\delta^+_s)\over 2+\delta^+_s-\delta^-_{s+1}}\rt)^{-2}\rt\}. \label{240}
 \eeq
 
 In the case of scattering objects $\bZ=\bX\bPsi^*$ with
 $\bPsi^*$ satisfying the RIP (\ref{rip}) 
 the thresholding rule (\ref{240}) holds 
 under  the following bound
 on the noise-to-scatterer ratio (NSR)  
 \beq
 \label{203'}
{\ep\over \xi_{\rm min}}  < 
 \sqrt{{(1+\delta^+_s)^2}{\xi^2_{\rm max}\over \xi^2_{\rm min}} +(1-\delta^-_s)^2\Delta}
 -{(1+\delta^+_s)}{\xi_{\rm max}\over \xi_{\rm min}}
 \eeq
 where  
$\xi_{\rm max}/\xi_{\rm min}$ is the dynamic range of
 scatterers.

 \end{theorem}
 
  \begin{proof}
  By  (\ref{211}) and  (\ref{226})
 \beq
 \label{204}
\rho^2+2\rho{\zeta_{\rm max}\over \zeta_{\rm min}}\sqrt{1+\delta^+_s} 
 < \Delta,\quad \rho={\ep\over \zeta_{\rm min}}
 \eeq
 implies  (\ref{202}) in Theorem \ref{thma}.  
 The sufficiency of (\ref{100}) now follows from
 solving the quadratic inequality (\ref{204}) for  $\rho$. 
 
 The derivation of the thresholding rule (\ref{240}) under the stronger condition (\ref{100})  is exactly
 the same as that of (\ref{241}). Alternatively, we can use (\ref{81}) and (\ref{82}) to verify  validity of
 the thresholding rule (\ref{240}) as follows. Let
 \[
 L=1-{\delta^-_{s+1}(1+\delta^+_s)\over 2+\delta^+_s-\delta^-_{s+1}}. 
 \]
By using (\ref{82}) and (\ref{ineq})  it is straightforward to check  that $
 2L^{-2}$ is greater than the right hand side of (\ref{82}). On the other hand, (\ref{202}) and Lemma \ref{prop3'} imply that
$2L^{-2}$ is smaller than the right hand side of (\ref{81}). 
 
 The proof for the case of scattering objects is exactly the same
 as above. 
 \end{proof}

\begin{remark}
\label{rmk72}
The right hand side of (\ref{203'}) decreases as the ratio
\beq
\label{super}
{(1-\delta_s^-)^2\Delta\over 
2(1+\delta_s^+)\xi_{\rm max}/\xi_{\rm min}}
\eeq
decreases. 
In the underresolved case (Section \ref{sec:num}),
$\delta_s^-$ is close to $1$, making the ratio
(\ref{super}) a small number. For a noise-to-scatterer
ratio  smaller than (\ref{super})
the $s$ scatterers can be perfectly localized by
the MUSIC algorithm with thresholding. This
is the superresolution effect. 

\end{remark}

A simple upper bound for  the RIC can be given in terms of the notion of coherence parameter $\mu(\tilde\bPhi)$ defined as
\[
\mu(\tilde\bPhi)=\max_{i\neq j} {\lt|\sum_{l} \tilde\Phi_{li}\tilde\Phi^*_{lj}\rt|\over
\sqrt{\sum_{l}|\tilde\Phi_{li}|^2\sum_{l} |\tilde\Phi_{lj}|^2}} . 
\]
Namely, $\mu(\tilde\bA)$ is the maximum of cosines of angles
between any two columns. 

The proof of the following well known result is elementary
and instructive. 
\begin{proposition} 
\label{prop2} For any $r\in \IN$, we have 
\[
\delta^\pm_r\leq \mu(\tilde\bA) (r-1).
\]

\end{proposition}
\begin{proof}
Calculating  
the quantity $\|\tilde\bA Z\|_2^2-\|Z\|_2^2$,  we have
\beqn
\lt|\|\tilde\bA Z\|_2^2-\|Z\|_2^2\rt|&=
\lt|\sum_{i\neq j}\tilde\Phi_i^*\tilde\Phi_j Z^*_iZ_j\rt|
&\leq \mu(\tilde\bA) \sum_{i\neq j}\lt| Z^*_i Z_j\rt|.
\eeqn
Using the quadratic inequality $2ab\leq a^2+b^2$ we obtain
\beqn
\sum_{i\neq j}\lt| Z^*_i Z_j\rt|&\leq 
{1\over 2} \sum_{i\neq j}\lt( |Z_i|^2+ |Z_j|^2\rt)
&\leq \sum_{i\neq j\atop Z_j\neq 0} |Z_i|^2\leq (r-1) \|Z\|_2^2.
\eeqn
Therefore (\ref{rip}) is satisfied with 
$\delta^\pm_{r}\leq \mu(\tilde\bA) (r-1). $

\end{proof}
\begin{remark}
For $s=2$  it follows from Proposition \ref{prop2} that
\[
\delta^\pm_2\leq \mu(\tilde\bA). 
\]
Since $\mu$ is almost surely less than unity for
randomly selected sampling directions, 
the MUSIC algorithm will find the true location
of object in the absence of noise, if there is only one object.

\end{remark}

The  coherence  bound for the most general setting 
of random sampling directions is this.
 \begin{proposition}\cite{cis-simo}\label{thm1} Suppose
 any two points in $\cK$ are separated by at least
 $\ell>0$. 
Let $\hat\bs_k, k=1,...,n$ be independently drawn from the distribution $f^{\rm s}$
on the $(d-1)$-dimensional sphere
independently and identically.  
Suppose
\beq
\label{m-2'}
N\leq {\alpha\over 8} e^{K^2/2}
\eeq
for any positive constants $\alpha, K$. 
Then   $\tilde\bA$  satisfies the coherence bound
\beqn
\label{mut'}
\mu(\tilde\bPhi) < \chi^{\rm s}+{\sqrt{2}K\over \sqrt{n}}
\eeqn
 with probability greater than $(1-\alpha)^2$
 where  $\chis$ satisfies the bound
 \beq
 \label{142''}
 &\chi^{\rm s}\leq {c_t}{(1+\om \ell)^{-1/2}} \|\pdfs\|_{t,\infty},&d=2\\
 &\chi^{\rm s}\leq {c_1}{(1+\om \ell)^{-1}} \|\pdfs\|_{1,\infty},&d=3\label{143''}
 \eeq
 where $\|\cdot\|_{t,\infty}$ is the H\"older norm
 of order $t>1/2$ and the constant $c_t$  depends only on $t$.
\end{proposition}
\begin{remark}\label{rmk3}

Replacing  $\tilde\bA$, $\hat\bs_k$ and $f^{\rm s}$ in Proposition \ref{thm1} by
 $\tilde\bPsi, \bd_k$ and $f^{\rm i}$, respectively, we have 
 the same conclusion about $\tilde\bPsi$. 

\end{remark}
The constraint  (\ref{m-2'}) on the 
number of  search points in the computation grid $\cK$
is relatively weak and 
allows an extremely refined grid.
However, to have a small $\chi^{\rm s}$
$\ell$ can not be small compared to the wavelength.

 Suppose  $\om\ell \geq C^2n $ for $d=2$ or
 $\om\ell\geq C\sqrt{n}$ for $d=3$ where
  \beq
 \label{120}
  C\geq {c_1\over \sqrt{2} K} \max\lt\{\|f^{\rm s}\|_{1,\infty}, \|f^{\rm i}\|_{1,\infty}\rt\}. 
  \eeq
  Then, according to Propositions \ref{thm1},  \ref{prop2} and Remark \ref{rmk3}, with high probability,  for 
  continuously differentiable distributions $f^{\rm s}, f^{\rm i}$
  we have  
  \beq
  \label{230}
\delta^\pm_{s}\leq  \delta^\pm_{s+1} \leq 2\sqrt{2}K s/\sqrt{n}. 
    \eeq 
Theorem \ref{cor2.1} then implies the following.
\begin{corollary}
\label{cor1}  Suppose  $\om\ell \geq C^2n $ for $d=2$ or
 $\om\ell\geq C\sqrt{n}$ for $d=3$  with $C$ given by (\ref{120}). 

Suppose that $\sqrt{n}/{s}\geq 4\sqrt{2} K$  (hence $\delta^\pm_s, \delta^\pm_{s+1} \leq  1/2$ by (\ref{230})) and that the {\rm NSR} obeys
\beq
\label{150}
{\ep\over \xi_{\rm min}}  < \lt(
 \sqrt{{9\over 4} {\xi^2_{\rm max}\over  \xi^2_{\rm min}}+{\Delta\over 4}}
 -{3\over 2} {\xi_{\rm max}\over \xi_{\rm min}}\rt)
 \eeq
 where 
 $\Delta$ is given in (\ref{202}). 
 Then  under the assumptions of Proposition \ref{thm1}
the MUSIC algorithm with the thresholding rule
 \beq
 \label{75'}
 \lt\{ \br\in \cK: J^\ep(\br) \geq  {128\over 25} \rt\}
  \eeq
recovers  exactly the locations of $s$ scatterers with probability at least $(1-\alpha)^2$. 
\end{corollary}

The value $128/25$ is arrived  from the fact that 
\[
\max_{\delta^+_s,\delta^-_{s+1}<1/2}
{\delta^-_{s+1}(1+\delta^+_s)\over 2+\delta^+_s-\delta^-_{s+1}}
\leq 3/8.
\]

\commentout{
For the near-field measurement,
the columns  take the form of
the Green function vector
\[
\phi_{\br_j}=\lt( G(\bs_1, \br_j),\cdots, G(\bs_n, \br_j)\rt)^T
\]
where $\bs_k$ are the locations of the sensors. 

 \begin{proposition}\cite{cis-simo}
 \label{thm2}
Suppose  
\beq
\label{m}
m\leq {\alpha\over 2} e^{2K^2/r_0^2},\quad \alpha>0
\eeq
where $c_0$ depends on the minimum distance $\Delta_{\rm min} $ between
$\{z=0\}$ and the lattice (For $d=2$, $r_0=
\cO(-\log{\Delta_{\rm min}})$; for $d=3$, $r_0=\cO(\Delta_{\rm min}^{-1})$).

The mutual coherence obeys
\beq
\label{1043}
\mu(\bPhi)& \leq & |G(\Delta_{\rm max})|^{-2} \lt({\sqrt{2}K \over \sqrt{n}}+ {c\over \red \sqrt{\om L}}\rt),\quad d=2\\
\mu(\bPhi) &\leq  & |G(\Delta_{\rm max})|^{-2} \lt({\sqrt{2}K \over \sqrt{n}}+ {c\over \red {\om L}}\rt),\quad d=3\label{1044}
\eeq
for some constant $c$ (independent of $\om>0$ for $d=2$  and $\om>1$ for $d=3$),   with probability greater 
than $(1-\delta)^2$,  where $\Delta_{\rm max}$
 is the largest
distance between the array and the lattice.  
\end{proposition}
}

\section{Planar objects: optimal recovery}\label{sec:opt}

Let us consider the favorable imaging geometry where all the scatterers lie on the
transverse plane $z=0$. Furthermore,
we consider the idealized situation where the locations of the scatterers are a subset $\cS$ of
a finite square lattice $\cK$ of spacing $\ell$
\beq
\label{44'}
\cK=\lt\{\br_j: j=1,...,N\rt\}=\lt\{(p_1\ell , p_2\ell, 0): p_1, p_2=1,...,\sqrt{N}\rt\},\quad j=(p_1-1)\sqrt{N}+p_2. 
\eeq
Hence the total number of grid points
 $N$ is  a perfect square.
 
 Suppose we choose the frequency such that
 \beq
 \label{freq}
 \om \ell=\sqrt{2}\pi.
 \eeq
 Let  $\ba_k=(\xi_k,\eta_k), k=1,...,n$ be  independently and  uniformly 
distributed random variables in  $[-1, 1]^2$ and
set
\beq
\label{ang}
\hat\bs_k={1\over \sqrt{2}}(\ba_k, \sqrt{2-|\ba_k|^2}).
\eeq
Let the incident directions $\bd_l, l=1,...,m$ be
selected the same way but independently from
$\hat \bs_k, k=1,...,n$.  It can be proved
that with $m\geq s$  the corresponding sensing matrix $\bPsi$
has rank $s$ with probability one. 

With (\ref{freq})-(\ref{ang}) and $j=(p_1-1)\sqrt{N}+p_2$
the scattering amplitude (\ref{4'}) for  linear extended objects  yields the following extended sensing matrix
\beq
\label{1001}
\tilde\Phi_{k,j}= e^{-\pi i \ba_k\cdot\bp} \in \IC^{n\times N}
\eeq
The matrix (\ref{1001}) is often referred to as the random partial
Fourier matrix in compressed sensing theory. 

The following is a standard result about random partial Fourier matrix \cite{Rau}. 
\begin{proposition} \cite{Rau}
\label{rau} Suppose $\ba_j, j=1,...,n$ are independently and  uniformly 
distributed in  $[-1, 1]^d, d\geq 1$. 
If 
\beq
\label{rau2}
{n\over \ln{n}}\geq C \delta_*^{-2} s\ln^2{s} \ln{N} \ln{1\over \gamma}
\eeq
for $\gamma\in (0,1)$ and some absolute constant $C$,  then 
with  probability at least $1-\gamma$
the random partial Fourier  matrix defined by (\ref{1001})
satisfies
the RIC bound
\beq
\label{rau22}
\delta^\pm_s<\delta_*. 
\eeq
\end{proposition}
\begin{remark}\label{rmk8}
The result holds true for sampling points $\ba_j$ 
which are i.i.d. uniform r.v.s in the discrete set
\beq
\lt\{\lt({k_1\over N^{1/d}}, {k_2\over N^{1/d}},\cdots,{k_d\over N^{1/d}}\rt): k_1, k_2=-N^{1/d},...,N^{1/d}-1\rt\}
\eeq
instead of $[-1,1]^d$ where $N^{1/d}$ is assumed to be an integer.
\end{remark}
\commentout{
\begin{remark}
Since $\bD_1$ and $\bD_2$ are isometries of
their respective vector spaces, $\tilde\bA$ satisfies
the RIC bound (\ref{rau22}) with the same probability
under (\ref{rau2}).  
\end{remark}
}

Assume for simplicity the plane wave incidence as before.
Choosing
 $\delta_*=1/2$ for sparsity $s+1$ in Proposition \ref{rau} and
 using Theorem \ref{cor2.1} we obtain the following result. 
\begin{corollary}\label{cor2} Suppose that (\ref{freq}) is true
and that  $\hat\bs_k=\bd_k, k=1,...,n$ with 
\beq
\label{rau3}
{n\over \ln{n}}\geq 4 C (s+1)\ln^2{(s+1)} \ln{N} \ln{1\over \gamma}. 
\eeq
If  the {\rm NSR} obeys (\ref{150})
the MUSIC algorithm with the thresholding rule (\ref{75'})  recovers  exactly the
locations of $s$ scatterers with probability at least $1-\gamma$. 

\end{corollary}

\subsection{Paraxial regime}

Here  we would like to extend the scattering  problem 
to the paraxial regime for the preceding set-up  where, instead of $n$ sampling directions, $n$ point sensors located on the transverse plane $z=z_0$
measure the scattered field. 

We shall make the paraxial approximation
for the Green function between the object plane 
$z=0$ and the sensor plane $z=z_0$:
\beq
\label{70'}
G(\bs,\br)={e^{i\om z_0}\over 4\pi z_0} 
e^{i\om (x^2+y^2)/(2z_0)} e^{-i\om x\xi /z_0}
e^{-i\om y\eta/z_0}
e^{i\om (\xi^2+\eta^2)/(2z_0)}
\eeq
with $\bs=(\xi,\eta,z_0), \br=(x,y,0)$. 
Denote
\beqn
G_p(\bs,\br)=e^{i\om|x-\xi|^2/(2z_0)}e^{i\om|y-\eta|^2/(2z_0)}.
\eeqn

Let $\bs_k=(\xi_k,\eta_k,z_0), k=1,...,n$ be
the locations of the transceivers. 
We have $\tilde\bA=\tilde\bPsi$ where the extended  matrix  $\tilde\bA$ is given
by \beq
\label{27}
\tilde\Phi_{k,l}=\lt( G(\bs_1, \br_j),\cdots, G(\bs_n, \br_j)\rt)^T,\quad j=1,...,N
\eeq
where $\br_j\in \cK$ defined in (\ref{44'}). 
After proper normalization  the extended sensing matrix $\tilde\bA$ can be written
as the product of three matrices
\beq
\label{52}
\tilde\bA=\bD_1 {\mathbf A}\bD_2
\eeq
where 
\[
\bD_1=\hbox{diag} (e^{i\om (\xi_j^2+\eta_j^2)/(2z_0)}) \in \IC^{n\times n},\quad \bD_2=\hbox{diag} (e^{i\om (x_l^2+y^2_l)/(2z_0)})\in \IC^{N\times N}
\]
 are  unitary and  
\beqn
{\mathbf A}={1\over \sqrt{n}} \lt[e^{-i\om \xi_j x_l /z_0}e^{-i\om \eta_j y_l /z_0}\rt] \in \IC^{n\times N}. 
\eeqn

Now suppose that  $(\xi_j,\eta_j), j=1,...,n$ are independently and  uniformly 
distributed in  $[-{A\over 2}, {A\over 2}]^2$ and write 
$ (\xi_j,\eta_j)= \ba_j\cdot A/2$ with
$\ba_j\in [-1,1]^2, j=1,...,n$. 
Write also $\br_l=(x_l,y_l,0)=(\bp\ell,0)$ where $\bp\in \IZ^2$. 
Then with
\beq
\label{102}
{A\ell\over \lambda z_0}=1
\eeq
${\mathbf A}$ takes the form
\beq
\label{1001'}
{\mathbf A}={1\over \sqrt{n}} \lt[e^{-\pi i \ba\cdot\bp}  \rt] \in \IC^{n\times N}
\eeq
which is exactly 
 the random partial
Fourier matrix  given in (\ref{1001}). Here and below $\lambda$ denotes the wavelength. 

Since both $\bD_1$ and $\bD_2$ are unitary and diagonal,
they leave both the $\ell^2$-norm and the sparsity
of a vector unchanged. Therefore Proposition \ref{rau}
and Remark \ref{rmk8}
are applicable  to $\tilde\bA$ given in (\ref{52}). 

Analogous to Corollary \ref{cor2} we have 
\begin{corollary}\label{cor2'} Suppose that (\ref{102}) is true
and that  there are $n$ transceivers satisfying (\ref{rau3}).
If the {\rm NSR} obeys (\ref{150}), then 
the MUSIC algorithm with the thresholding rule (\ref{75'})  recovers exactly the
locations of $s$ scatterers  with probability at least $1-\gamma$. 
\end{corollary}
In the paraxial setting (\ref{102}) in  Corollary \ref{cor2'} replaces the
condition (\ref{freq}) in Corollary \ref{cor2}. Condition
(\ref{102}) is exactly the classical Rayleigh criterion for
resolution which is, in this case, the grid spacing $\ell$.

\section{Comparison with basis pursuit}\label{sec:bp}
\commentout{
In compressed sensing, the linear relation (\ref{20'}) between
$\bX$ abd $\bY$ is often written differently as follows: Let $Y=\hbox{\rm vec}(\bY)\in \IC^{nm}$
be the vectorized version of the data matrix $\bY\in \IC^{n\times m}$ where $Y$ is obtained by stacking
the $m$ column vectors of $\bY$ into a $nm$-dimensional  vector. 
Let   $\tilde X\in \IC^{N}$ be the vector  containing the diagonal elements of $\tilde\bX$. Then (\ref{15'}) can
be written as
\beq
\label{60}
Y={\mathbf A} \tilde X
\eeq
for some matrix ${\mathbf A}\in \IC^{nm\times N}$. 

}

\commentout{
The advantage of
For simplicity of notation and ease of
comparison, we focus on the case
of $m=1$, i.e. $\tilde\bPsi=\tilde\Psi\in \IC^N$. In this case, we can write
\[
Y=\tilde \bA  Z \in \IC^n,\quad Z=\tilde \bX \tilde \Psi^*\in \IC^N.
\]
If $Z$ can be recovered, then by component-wise division 
$\tilde X$ can be also recovered since $\Psi$ vanishes
nowhere. Note that $\tilde X$ contains all the information
about the locations and amplitudes of the scatterers.

Reconstruction  in compressed sensing
is given by the minimizer, if  unique, of 
the $L^1$-based minimization principle
\beq
\label{BP}
\min_{Z'\in \IC^N} \|Z'\|_1,\quad \hbox{s.t.}\,\, Y=\tilde\bA Z'
\eeq
called 
the basis pursuit (BP) \cite{BDE}. 
}

In the standard compressed sensing theory, one usually considers
the following data model
\beq
\label{43}
Y^\ep=\tilde \bA  Z +E\in \IC^n,\quad Z\in \IC^N, 
\quad \|E\|_2\leq\ep
\eeq
where the data and the object are  vectors 
and employs
the relaxed minimization principle called the  Basis Pursuit Denoising (BPDN)
 \cite{BDE, CDS} 
\beq
\label{relax2}
\min_{Z'\in \IC^N} \|Z'\|_1,\quad \hbox{s.t.}\,\, \|Y^\ep-\tilde\bA Z'\|_2\leq \epsilon
\eeq
for reconstruction. The noiseless version  $\ep=0$ of 
(\ref{relax2}) is called the {\em Basis Pursuit} (BP). 
Note that  BPDN  uses only {\em one column} 
of the  MUSIC model (\ref{20'}).

When $Z$ is not $s$-sparse, consider the
best $s$-sparse approximation $Z^{(s)}$ of $Z$. 
Clearly, $Z^{(s)}=Z$ if $Z$ is $s$-sparse. 

Denote the BPDN minimizer by $\hat Z$.  When does $\hat Z$ give a good approximation to the
true $Z$? Again, the RIP (\ref{rip}) gives a useful
characterization \cite{Can}. 

\commentout{
The first one  is in terms of the coherence parameter $\mu(\tilde\bA)$.  
\begin{proposition}\cite{DET}
\label{prop6} Suppose $\tilde\bA$ has all unit columns. 
Suppose the sparsity $s$  of the target vector $X$  satisfies
\[
s< {1\over 4} \lt(1+{1\over \mu(\tilde {\mathbf \Phi})}\rt). 
\]
Then the BPDN  minimizer $\hat Z$ is unique and obeys the
error bound
\[
\|\hat Z-Z\|_2\leq {2\ep\over \sqrt{1-(4s-1)\mu(\tilde {\mathbf \Phi})}}.
\]
\end{proposition}
}

\begin{theorem} \label{prop4} 
Suppose the  RIC  of $\tilde\bA \in \IC^{n\times N}$
satisfies the inequality 
\beq
\label{ric}
{\sqrt{2}\over 2} \delta_{2s}^++\lt({\sqrt{2}\over 2}+1\rt)
\delta_{2s}^-<1
\eeq
 Then the BPDN  minimizer  $\hat Z$ is unique and  satisfies
 the error bound
 \beq
 \|\hat Z-Z\|_2&\leq & C_1s^{-1/2}\|Z-Z^{(s)}\|_1+C_2{\ep}\label{est}
 \eeq
 where
 \beqn
 C_1&=&{2+(\sqrt{2}-2)\delta_{2s}^-+\sqrt{2}\delta_{2s}^+ \over 1- {\sqrt{2}\over 2} \delta_{2s}^+-\lt({\sqrt{2}\over 2}+1\rt)
\delta_{2s}^- }\\
 C_2&=&
{4\sqrt{1+\delta^+_{2s}}\over 1- {\sqrt{2}\over 2} \delta_{2s}^+-\lt({\sqrt{2}\over 2}+1\rt)
\delta_{2s}^- }. 
\label{c2}
 \eeqn
\end{theorem}
\begin{remark}
The real-valued version with $\delta_s=\delta^\pm_s$ of Theorem \ref{prop4} is proved
in \cite{Can}. The proof for the complex-valued setting 
follows the same line of reasoning with
minor modifications. For reader's convenience and for
the purpose of showing where adjustments are needed,
 the full proof for the complex-valued setting  is given in the Appendix \ref{app:B}.  
 \end{remark}
  \begin{remark}
  Theorem \ref{prop4} does not guarantee exact recovery
  of support when $E\neq 0$. 
  
  An alternative approach to BPDN  is greedy algorithms such as
  the orthogonal matching pursuit (OMP). 
  The exact recovery of support by OMP is established
  for any $s$-sparse object vector $Z$ such that 
  the noise-to-object ratio satisfies 
  \beq
  {\ep\over Z_{\min}}\leq {1\over 2} +\mu(\tilde\bA)  ({1\over 2}-s)
  \label{260}
  \eeq
  where $Z_{\rm min}$ is the smallest absolute nonzero component of $Z$ \cite{DET}. In order for the right hand side
  of (\ref{260}) to be positive, it is necessary that
  \beq
  \label{261}
  s<{1\over 2}+{1\over 2\mu(\tilde\bA)}.
  \eeq
  
  \end{remark}
  \begin{remark}
  \label{rmk6}
 In the case of planar  objects, under the assumptions  of Proposition \ref{rau} (with $\delta_*=\sqrt{2}-1$), BP yields the exact solution $\hat Z=Z$ for
 $n=\cO(s)$ sampling directions (or sensors)  and just one
 incident wave, modulo logarithmic factors. In comparison, the performance guarantee for MUSIC in  Corollary \ref{cor2} assumes $n=\cO(s)$ sampling {\em and}  incident directions. 
 \end{remark}

Using Proposition  \ref{prop2} and Theorem \ref{prop4}, we obtain
\begin{corollary}\label{cor6}
If
\beqn
s<{1\over 2}+{\sqrt{2}-1\over  2\mu(\tilde\bA)},
\eeqn
cf. (\ref{261}), then (\ref{est}) holds true. 
\end{corollary}

 \begin{remark}\label{rmk7}
Under the assumptions of  Proposition \ref{thm1}  with continuously differentiable $f^{\rm s}$ and $f^{\rm i}$   BP recovers
the $s-$sparse object exactly  $\hat Z=Z$ in the noiseless case $\ep=0$ for $n=\cO(s^2)$ and sufficiently high
frequency. 

This is similar to the performance guarantee for MUSIC in Corollary \ref{cor1}. However,  the performance guarantee for MUSIC assumes 
$\cO(s^2)$ incident waves  while
the performance guarantee for  BP assumes only {\em one} 
incident wave. 

\end{remark}

\section{Spectral estimation and {source  localization}}\label{sec:spec}
Let us turn to the original application where the MUSIC algorithm
arises, namely the source localization and
the  frequency estimation for  multiple random signals. 
The two applications share almost exactly the same mathematical formulation.  

Suppose the random signal $x(t)$ consists of random linear
combinations of 
$s$ time-harmonic components from
the set
\[
\{e^{-i2\pi \om_j t}: \om_j ={j\over N}, j=1,...,N\}.
\] 
Let us write
\beq
\label{50}
x(t)=\sum_{j=1}^{N} a_j e^{-i2\pi \om_j t}
\eeq
and assume that there is a fixed set $\cS$ (i.e. deterministic support) of $s$ nonzero amplitudes and the elements in the complementary
set $\cS^c$ are zero almost surely.
 
 Consider the noisy signal model
 \beq
 \label{40}
 y(t)=x(t)+e(t)
 \eeq
 where $e(t)$ is the Gaussian white-noise. 
 The question is to find out which $s$ components
 are non-zero by sampling $y(t)$.  
 
 Consider random sampling times $t_k, k=1,...,n$ which
 are i.i.d. uniform r.v.s in the set $\{1,...,N\}$. 
 Write $Y=(y(t_k))\in \IC^n, E=(e(t_k))\in \IC^n$ and $Z=(a_j)\in \IC^N$. 
 Then by (\ref{40}) we have
 \beq
 \label{41}\label{u1}
 Y^\ep&=&\tilde \bA Z+E\\
 \Phi_{k,j}&=&{1\over \sqrt{n}} e^{-i2\pi t_kj/N}\in \IC^{n\times N}\nn
 \eeq
 cf.  (\ref{43}).  From the one-dimensional setting of Proposition \ref{rau}
 and Remark (\ref{rmk8}) we know that
 if (\ref{rau2}) is satisfied with $\delta_*=\sqrt{2}-1$
 then the RIC of $\tilde\bA$ obeys the bound (\ref{rau22}) with
 probability at least $1-\gamma$. Applying  BPDN  to (\ref{41}) we obtain the
 error bound (\ref{est}) with $\cO(s)$ data, modulo logarithmic factors. 
 
 How does MUSIC perform in this case? The standard MUSIC proceeds as follows.  Let $\bR_Y=\IE[YY^*]\in \IC^{n\times n}$, 
 $\bR_Z=\IE[ZZ^*]\in \IC^{N\times N}$ and $\bR_E=\IE[EE^*]$ be the covariance matrices of
 $Y$, $Z$ and $E$, respectively. Note that
 $\bR_Z$ is sparse and has at most rank $s$.

 Suppose the noise and
 the signal are independent of each other. Then we have
 \beq
 \label{43'}
 \bR_Y-\bR_E=\tilde\bA \bR_Z\tilde\bA^*
 \eeq
 which is of the form (\ref{15'}). 
 
 Assume that $\bR_Z$ has rank $s$. This is true, for example when $a_i$ are zero-mean,  independent random variables. 
 In this case, $\bR_Z=\hbox{diag}(\IE|a_i|^2)$ has
 exactly $s$ nonvanishing diagonal elements. 
 
Let $\bX \in \IC^{s\times s}$ be the $\cS\times \cS$ submatrix of $\bR_Z$.  By assumption, $\bX$ has rank $s$. Let $\bA=\bPsi=\tilde\bA_{\cS}\in\IC^{n\times s}$ be the column submatrix restricted to 
the index $\cS$, the (deterministic)  support of $Z$.
Let $
\bY= \bR_Y-\bR_E\in \IC^{n\times n}. 
$
We then  can rewrite (\ref{43'}) in the form (\ref{31}) suitable for MUSIC. 

  Since  $\bX$ has full rank,  the  the ranges of
 $\bA$ and $\bY$
 coincide. To guarantee the exact
 recovery by MUSIC, it suffices to show that the RIC of $\tilde\bA$ satisfies the bound  $\delta^-_{s+1}<1$
 which follows from Proposition  (\ref{rau}) under the condition
 $n=\cO(s)$, modulo logarithmic factors, cf. (\ref{rau3}). 
 
 Let us formally state the performance guarantee for MUSIC as applied
 to the problem of spectral estimation.
 
 \begin{corollary}\label{cor4} Suppose $\bR_Z$ has rank $s$ and
 suppose that the noise is independent from the signal. Assume
 that
 the number $n$ of time samples
 satisfies
 \beq
 \label{71'}
 {n\over \ln n}\geq C(s+1) \ln^2{(s+1)} \ln N \ln\gamma^{-1},
 \eeq
 cf. (\ref{rau3}). Then 
for any noise level  the singularities of the MUSIC
imaging function  (\ref{mus}) coincide with  the frequencies  in the random signals (\ref{50})
 with probability $1-\gamma$. 
 \end{corollary}
\begin{remark}\label{rmk9} The  number of time samples assumed here  is similar for both 
MUSIC and  BPDN. However,   many realizations of $Y$ and $Z$ are needed
  to calculate the covariance matrices accurately and form
  the equation (\ref{43'}) before
  the MUSIC reconstruction.  Once (\ref{43'}) holds with
  sufficient accuracy, then the noise  structure does not
  affect reconstruction as long as the noise
  is independent of the signal.
  
  In the absence of
  abundant realizations of signals, though,  BPDN is the preferred
  method for spectral estimation. Indeed, BPDN can
  identify the frequencies  approximately with
  just {\bf one}  realization of signals. The recovery error is at worst linearly proportional  to the noise level as  in (\ref{est}).   \end{remark}

The source localization  problem can be treated in the same vein
as follows. 

Let us assume that $s$  source points are distributed
in the grid $\cK$ defined in (\ref{44'})
 and 
  each  source point emits a signal
described 
by the paraxial Green function (\ref{70'})  times the source  amplitude 
 which is  recorded by the $n$ sensors located at $\ba_i, i=1,...,n$ in the plane $z=z_0$.

 Let $Z=(\xi(\br_j))\in \IC^N$.  After proper normalization, the data vector $Y$ can be written as (\ref{41})
with the sensing matrix $\tilde\bA$ of  the form
(\ref{52}). 

 By the same analysis as before we arrive at the conclusion
 \begin{corollary}\label{cor4'} Suppose $\bR_Z$ has rank $s$ and
 suppose that the noise is independent of the signals. 
 Let  the number $n$ of time samples
 satisfy (\ref{71'}). 
  For any noise level  the singularities of the  MUSIC
 imaging function (\ref{mus}) coincide with the source locations 
 with probability $1-\gamma$.  
 \end{corollary}

\section{Resolution and grid spacing}\label{sec:grid}
Being essentially a gridless method, MUSIC's flexibility with grid spacing is an advantage that the current BPDN-based imaging methods do not yet possess.

Let $\ell$ a length scale to be determined below
and let $\cS_\ell=\{\br\in \cK: \hbox{\rm dist}(\br,\cS)\leq \ell\}$ be the $\ell$-neighborhood
of the objects. For the problem of inverse scattering,  $\hbox{\rm dist}(\br,\cS)$ typically refers to the {\em physical }
or Euclidean  distance in the spatial domain. 
We would like to derive a thresholding rule
which can eliminate all false alarms
(i.e. artifacts)  occurring outside  $\cS_\ell$, no matter how refined the grid
spacing is relative  to the frequency.

Let $\tilde\bA$ be the extension of $\bA$ over
a fine grid of spacing $\tilde\ell$ which may be much smaller than $\om^{-1}$. When $\tilde\ell=0$, the computation domain
$\cK$ is a
continuum. 

Generalizing the definition (\ref{90}), we define 
\beq
\label{91}
\Gamma_\cS(\ell)\equiv \min_{\br\in \cS_\ell^c} \|\phi_\br\|_2^{-1} \|\cP\phi_\br\|_2=\sqrt{1-\max_{\br\in \cS_\ell^c} \|\phi_\br\|^{-2}\phi_\br^*\bA \bA^\dagger \phi_\br}.  
\eeq
Clearly,  $\Gamma_\cS=\Gamma_\cS(0^+)$. In the noiseless case, 
the exact recovery of $\cS$ by MUSIC is equivalent to
$\Gamma_\cS(0^+)<1$. 

Extending the proof of Lemma \ref{prop3'} to $\Gamma_\cS(\ell)$ we have
\begin{lemma}\label{lem3}
\beq
 \Gamma_\cS(\ell )\geq 1-\max_{ \cS'=\cS_\ell\cup \{\br\}\atop \br\in\cS_\ell^c
}{\delta^-_{\cS'}(1+\delta^+_{\cS})\over 2+\delta^-_{\cS}-\delta^+_{\cS'}}. 
\eeq
\end{lemma}

Extending  the analysis leading to (\ref{81})-(\ref{82})   we have
\beq
\label{92}
J^\ep(\br)&\geq&  {(1-2\rho)^4\over 4\rho^2(1-\rho)^2} ,\quad \br\in \cS\\
J^\ep(\br)&\leq&\lt( \Gamma^2_\cS(\ell)  -{4\rho^2(1-\rho)^2\over (1-2\rho)^4}\rt)^{-1},\quad \br\in \cS_\ell^c. \label{93}
\eeq
 
 The following result is analogous to Theorem \ref{thma}.
 \begin{theorem}\label{thm3} Suppose 
 $\delta^-_{\cS'}<1, \cS'=\cS\cup \{\br\}, \forall \br\in \cS^c_\ell$. 
 
 If 
 \beq
\label{242}
{\|\cE\|_2\over \sigma_{\rm min}}< \Delta_\ell={1\over 2}-{1\over 2}{1\over \sqrt{\sqrt{2}\Gamma_\cS(\ell)+1}}
\eeq
then 
\beq
\label{244}
\cS\subset \Theta=\lt\{\br: J^\ep(\br) \geq 2\Gamma^{-2}_\cS(\ell)\rt\}
\eeq
where $\Theta\cap \cS^c_\ell=\emptyset$. 

\commentout{
 the maximum of  $J^\ep$
 in $\cS_\ell^c$ is less than $2/\Gamma^2_{\cS}(\ell)$
 while $J^\ep(\br) > 2/\Gamma^2_{\cS}(\ell), \forall\br\in \cS$. In other words, the search in  $\{\br: J^\ep(\br) > 2/\Gamma^2_{\cS}(\ell)\}$
guarantees an  approximate recovery of support 
in $\cS_\ell$. }
\end{theorem}

\commentout{
 Hence with proper formulation
 the constraint $\om \ell\gg 1$ (cf. (\ref{42''})-(\ref{43''}))
 can be relaxed if 
one seeks only approximate, instead of exact,  recovery within  a $\ell$-neighborhood
of the objects. 
}

Lemmas \ref{lem3}, \ref{prop4'} and Theorem \ref{thm3}
then implies the following result analogous to Theorem \ref{cor2.1}. 
\begin{theorem}
\label{thm4} Suppose 
 $\delta^-_{\cS'}<1, \cS'=\cS\cup \{\br\}, \forall \br\in \cS^c_\ell$. 
If  the {\rm NSR} 
obeying the upper bound 
 \beq
 \label{80}
{\ep\over \xi_{\rm min}}  < 
 \sqrt{{(1+\delta^+_\cS)^2}{\xi^2_{\rm max}\over \xi^2_{\rm min}} +(1-\delta^-_\cS)^2\Delta_\ell}
 -{(1+\delta^+_\cS)}{\xi_{\rm max}\over \xi_{\rm min}}
 \eeq
  with $\Delta_\ell$ as in  Theorem \ref{thm3}
then  
\beq
\label{245}
\cS\subset \Theta=\lt\{\br: J^\ep(\br) \geq
2 \lt( 1-\max_{\cS'=\cS\cup\{\br\}\atop\br\in \cS_\ell^c}{\delta^-_{\cS'}(1+\delta^+_\cS)\over 2+\delta^+_\cS-\delta^-_{\cS'}}\rt)^{-2}\rt\}
\eeq
where $\Theta\cap \cS^c_\ell=\emptyset$.
\end{theorem}




Let us now give an estimate of the length scale $\ell$
for (\ref{80}) to be a useful upper bound for {\rm NSR}.
 Let us  focus on the general setting of 
Proposition \ref{thm1}, namely arbitrarily located scatterers and  random sampling directions. 

We resort to the following result analogous to Proposition \ref{prop2}. The proof is exactly the same as before and is omitted here. 
\begin{proposition}For any set $\cB\subset\cK, |\cB|\leq r$, we have
\[
\delta^\pm_\cB\leq \mu(\tilde\bA_\cB) (r-1).
\]
\label{prop10}

\end{proposition}
 
 To proceed,  let us tailor  the estimate in Proposition \ref{thm1} to
 the current setting as follows.
 
  \begin{proposition}\cite{cis-simo}\label{prop11}
  Suppose the physical distances between two points
  corresponding to any two members of 
  $\cB\subset \cK$ are at least $\ell$.
   Let $\hat\bs_k, k=1,...,n$ be independently drawn from the distribution $f^{\rm s}$
on the $(d-1)$-dimensional sphere
independently and identically.  
Suppose 
\beqn
\label{m-2}
|\cK|\leq {\alpha\over 8} e^{K^2/2}
\eeqn
for any positive constants $\alpha, K$. 
Then   $\tilde\bA_\cB$  satisfies the coherence bound
\beqn
\label{mut}
\mu(\tilde\bPhi_\cB) < \chi^{\rm s}+{\sqrt{2}K\over \sqrt{n}}
\eeqn
 with probability greater than $(1-\alpha)^2$
 where  $\chis$ satisfies the bound (\ref{142''})-(\ref{143''}). 
\end{proposition}

 Suppose  $\om\ell \geq C^2n $ for $d=2$ or
 $\om\ell\geq C\sqrt{n}$ for $d=3$ where $C$ is given
 by (\ref{120})
 and assume
 that the $s$ scatterers are separated by at least $\ell$ from
 one another.
  Then, according to Proposition \ref{prop11} and Proposition \ref{prop10}, with high probability,  for any
  continuously differentiable sampling distribution $f^{\rm s}$ 
  \[
\delta^\pm_{\cS}\leq  \delta^\pm_{\cS'} \leq 2\sqrt{2}K s/\sqrt{n}
    \]
  for all $\cS'=\cS\cup \{\br\}, \br\in \cS_\ell^c$. 
Hence we have the following analogous result to Corollary \ref{cor1}. 
\begin{corollary}
\label{cor10}  Suppose  $\om\ell \geq C^2n $ for $d=2$ or
 $\om\ell\geq C\sqrt{n}$ for $d=3$  with $C$ given by (\ref{120}). 

Under the assumptions of Proposition \ref{prop11} (for $\cB=\cS, \cS'$), $\sqrt{n}/{s}\geq 4\sqrt{2} K$  and  the NSR bound (\ref{150}),
\[
\cS\subset \Theta= 
 \lt\{ \br\in \cK: J^\ep(\br) \geq  {128\over 25} \rt\}
  \] 
 with probability at least $(1-\alpha)^2$ where $\Theta\cap \cS^c_\ell=\emptyset$.
 \end{corollary}


\section{Numerical tests}\label{sec:num}

In the simulations, $z_0=10000$, 
 $\lambda=0.1$ and  the search domain is $[-250, 250]^2$ with grid spacing $\ell=10$ on the transverse plane $z=0$.  The scatterers are independently and uniformly distributed  
on  the grid with amplitudes independently and uniformly
distributed   in the range  $[1, 2]$. 
The sensors are independently and uniformly  
distributed  in the domain $[-A/2,A/2]^2$ 
with various $A$. The source locations
are identical to the sensor locations. 
In the set-up,   condition (\ref{102}) is satisfied with  $A=100$.
With these parameters, 
the paraxial regime is about to set in (cf. \cite{cs-par}).
Note, however, that all the simulations are performed
with the exact Green function. 

\commentout{
We use the 
true Green
function in the computation of  scattered waves
and in the inversion step the exact Green function and its  paraxial approximation to construct the sensing matrix (for comparison). In other words, we allow model 
 mismatch between the propagation  and inversion steps.
 The degradation in performance can be seen in the figures 
 but is still manageable. 
 }

 In our simulations we have used   the Matlab  codes 
YALL1 (acronym for {\em Your ALgorithms for L1},  available at
{\tt http://www.caam.rice.edu/$~$optimization/L1/YALL1/}).
\commentout{
and 
SP (acronym for {\em Subspace Pursuit}, available at {\tt http://igorcarron.googlepages.com/cscodes}). 
}
YALL1 is  a  L1-minimization solver 
 based on the Alternating Direction Method \cite{YZ}.
 \commentout{
 while SP is a greedy algorithm which has a significantly
 lower computational complexity and,   under the condition  $\delta_{3s}<0.06$, is guaranteed to exactly recover  of $s$-sparse objects  via a finite number of iterations \cite{DM}.
}

\begin{figure}[t]
\begin{center}
\includegraphics[width=0.45\textwidth]{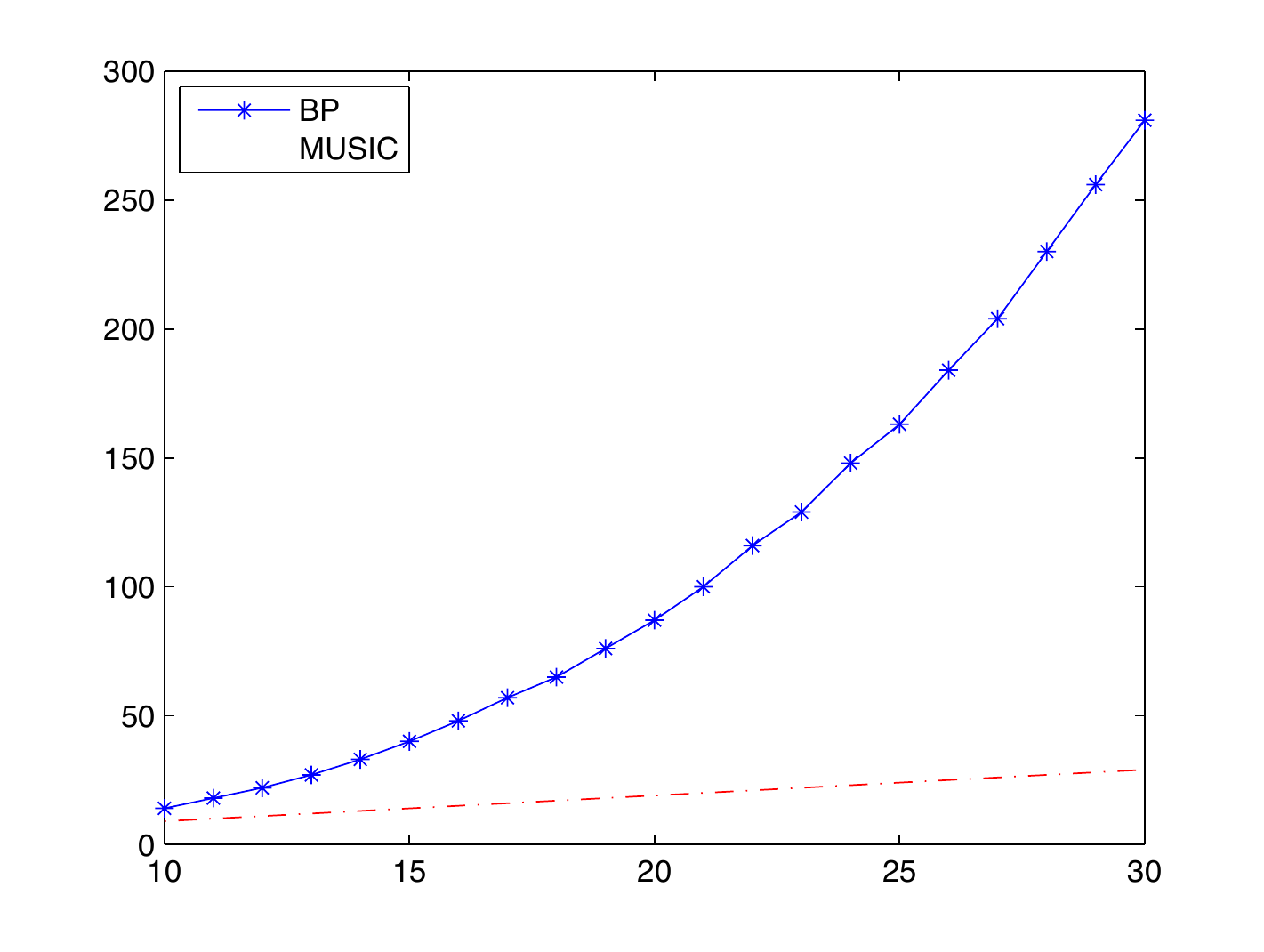}
\includegraphics[width=0.45\textwidth]{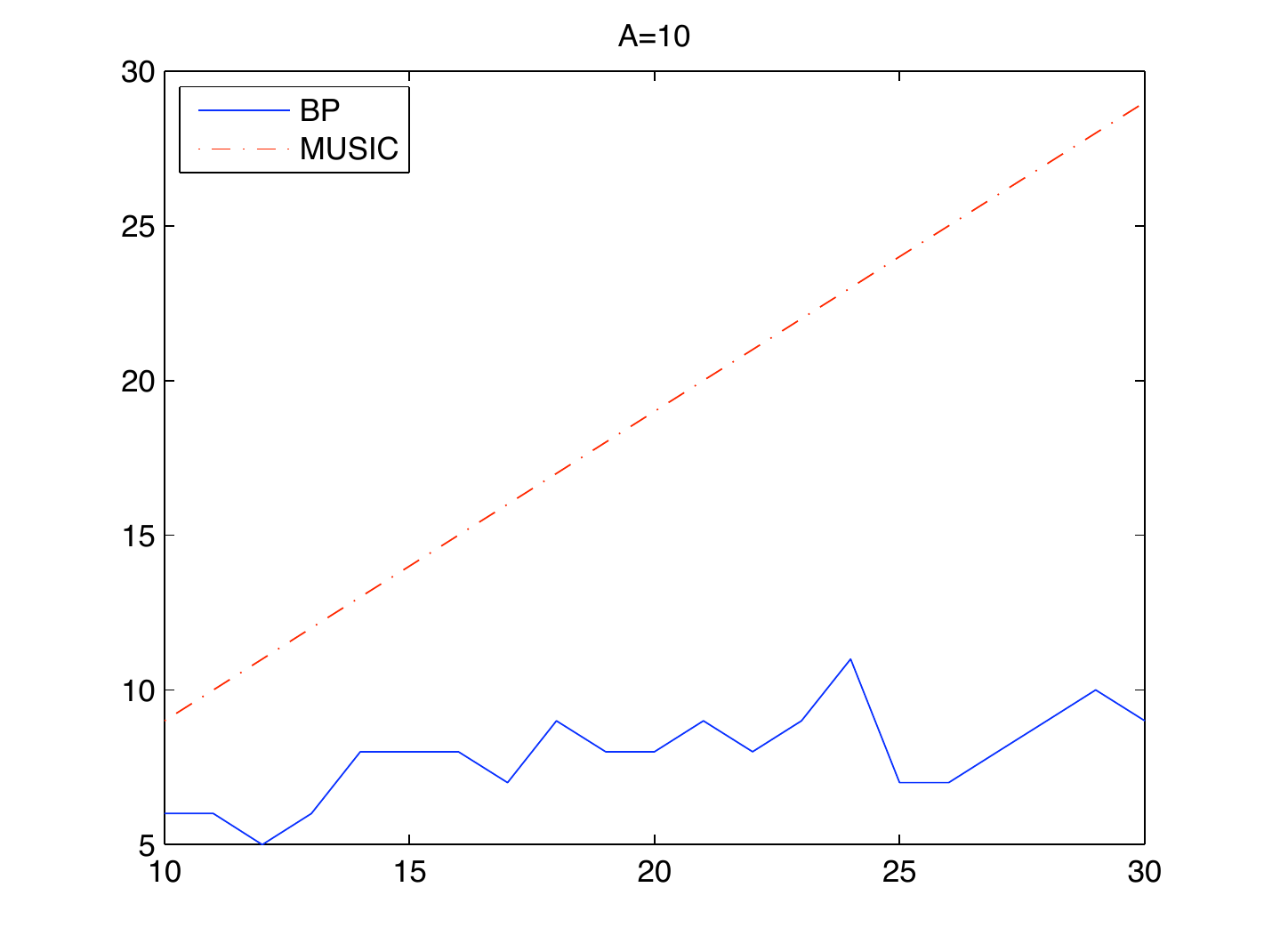}
\end{center}
\caption{Comparison of MUSIC and BP  performances, with both  using the whole data matrix: the number $s$ of recoverable scatterers versus the number of sensors $n$ 
with $A=100$ (left),  the well-resolved case,   and $A=10$ (right),
the under-resolved case. In the well-resolved  case,
BP delivers   a much better (quadratic-in-$n$) performance than MUSIC; in the under-resolved case, MUSIC outperforms BP whose performance tends to be unstable in this regime. The numbers
of recoverable scatterers by BP are calculated based on successful recovery of  at
least  90 out of 100 independent realizations of
transceivers and scatterers while the success rate of MUSIC is 100\%. }
\label{fig2}
\end{figure}
\begin{figure}[t]
\begin{center}
\includegraphics[width=0.45\textwidth]{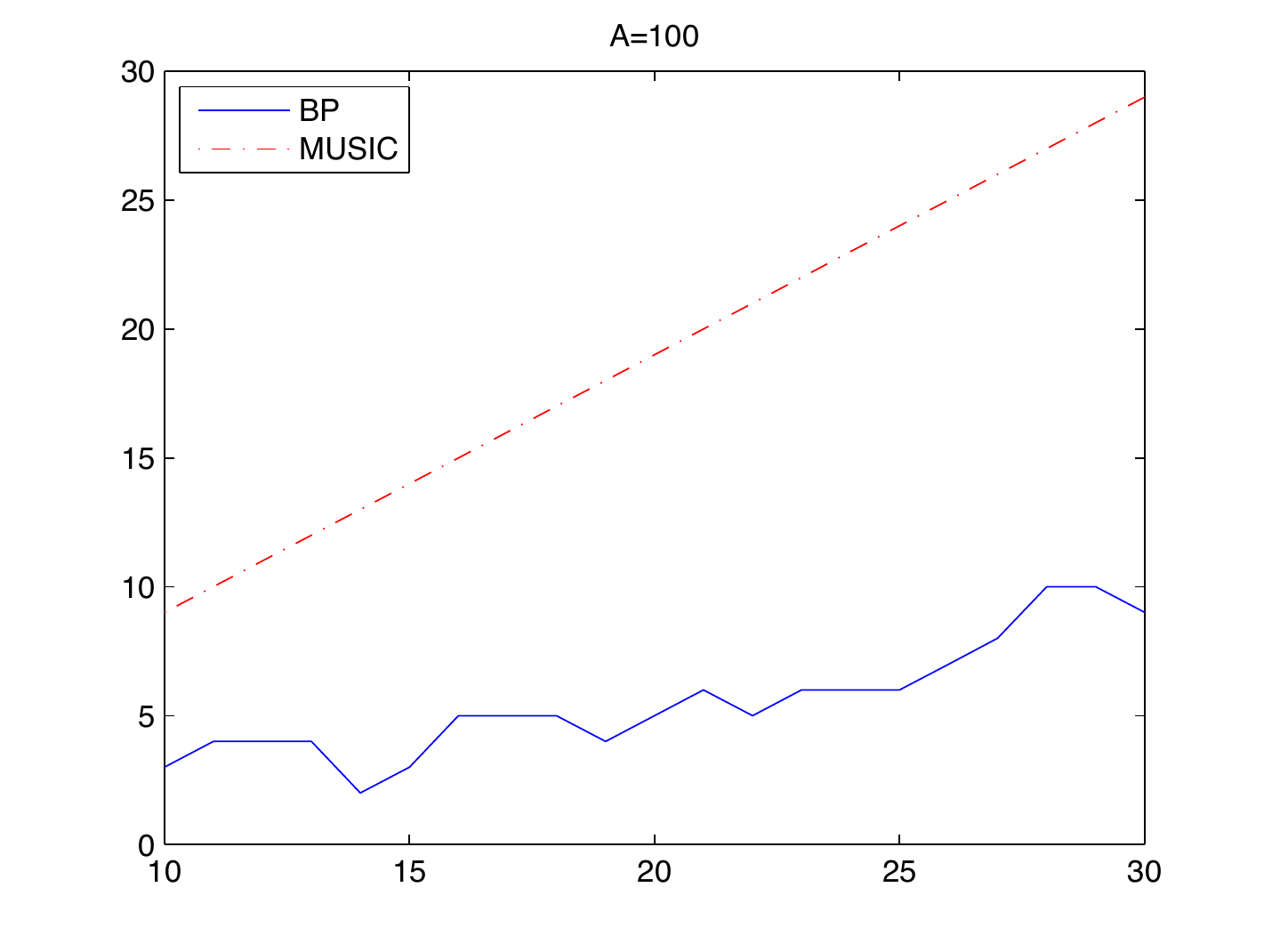}
\includegraphics[width=0.45\textwidth]{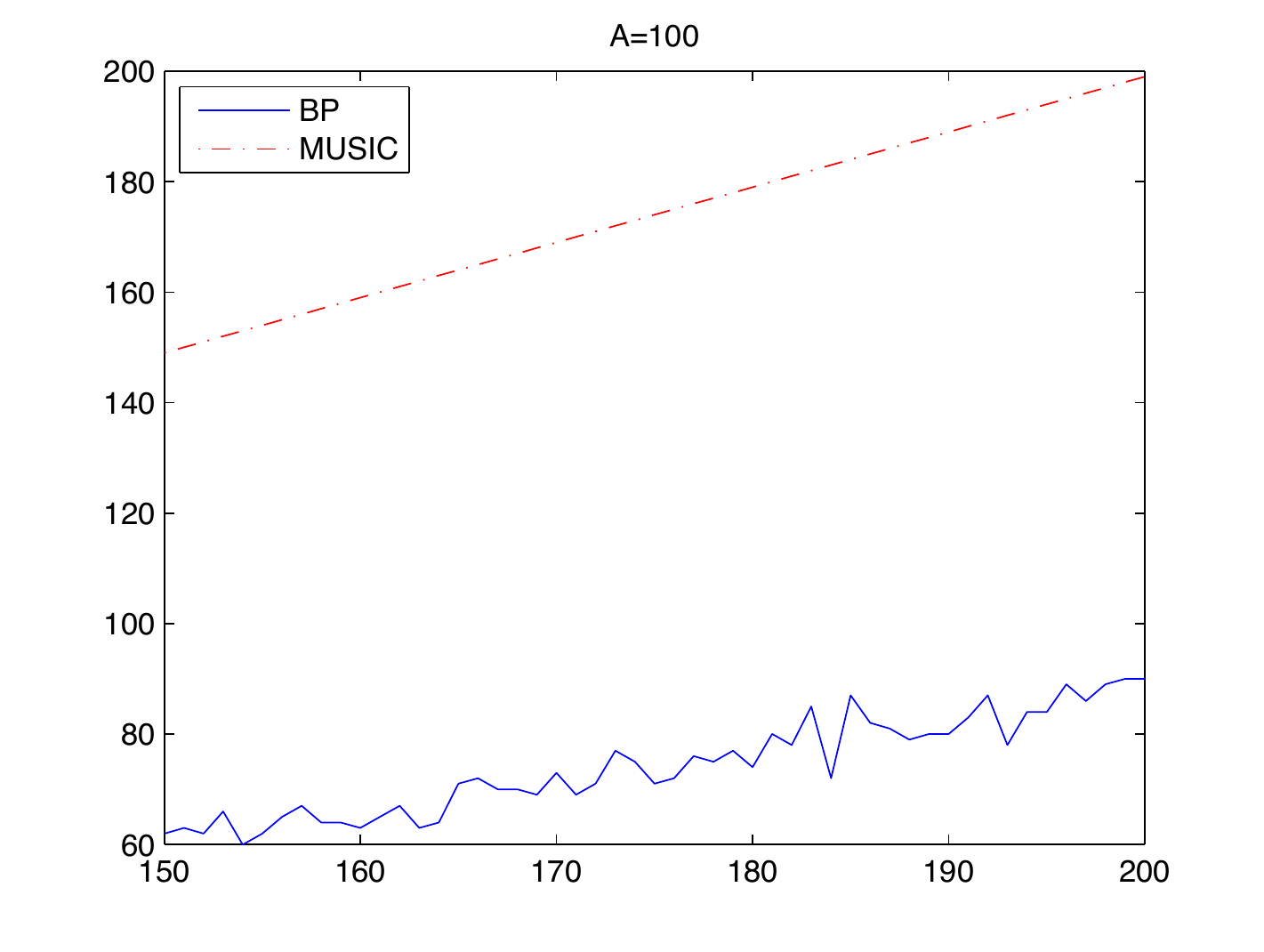}
\end{center}
\caption{Comparison of MUSIC and BP  performances  with BP employing  only {\em single}  column of the data matrix: the number $s$ of recoverable scatterers versus the number of sensors $n$ 
with $A=100$ for $n\in [10, 30] $ (left) and
$n\in [150, 200]$ (right). Both BP curves show a roughly linear behavior with slope less than that of the MUSIC curves.  }
\label{fig3}
\end{figure}

Figure \ref{fig2} compares the performances of MUSIC and
BP in the well-resolved case $A=100$  and the under-resolved
case $A=10$ where the aperture is only one tenth of
that satisfying   (\ref{102}). For Figure \ref{fig2}  BP is carried out on the data matrix  $\bY$ with
the sensors coincident with the sources, i.e. $\bA=\bPsi$. To put the problem
in the proper set-up for BP, we vectorize $\bY$ by
staking its $n$ columns and denote  the resulting $\IC^{n^2}$ vector by $Y$. We vectorize the diagonal matrix $\tilde\bX$ by listing its  $N$ diagonals as a $\IC^{N}$ vector $\tilde X$. 
The BP performance of this set-up has been 
analyzed  in \cite{cs-par}. The numbers of recoverable scatterers  shown in Figure \ref{fig2} are computed for at least $90\%$ recovery rate based
on $100$ independent realizations of
transceivers and scatterers. In both cases, MUSIC recovers $s=n-1$ scatterers with certainty. Clearly, for the well-resolved case,
 BP has a far superior performance to MUSIC. Indeed, it
can be shown that BP can recover $s=\cO(n^2)$ scatterers
with high probability in the well-resolved case \cite{cs-par}. 
The quadratic behavior is illustrated by the near-parabolic curve
in Figure \ref{fig1}(left).  For the under-resolved case, however, MUSIC outperforms BP by
a significant margin, Figure \ref{fig1}(right). 
As pointed out in Remark
\ref{rmk72}, the MUSIC algorithm 
has the superresolution capability for a sufficiently 
small noise-to-scatterer ratio.

If  only one column of $\bY$ is used in BP as discussed in Section \ref{sec:bp},
then MUSIC outperforms BP by a wide margin
even in the well-resolved case, 
Figure \ref{fig3}. 

\commentout{
\begin{figure}[t]
\begin{center}
\includegraphics[width=0.45\textwidth]{Active_SP.pdf}
\end{center}
 \caption{The empirical maximum number of recoverable 
{\em Born} scatterers  with $\xi \in [1,2]$ v.s. the number $n$
 of antennas. 
 The data for $n\in [10,30]$ in the MR plot is fitted with
  the parabola (blue-dashed curve): $-57.5400 +0.2964 n^2$. }
 \label{fig-scatt}
 \end{figure}
 }

\begin{figure}[t]
\begin{center}
\includegraphics[width=0.3\textwidth]{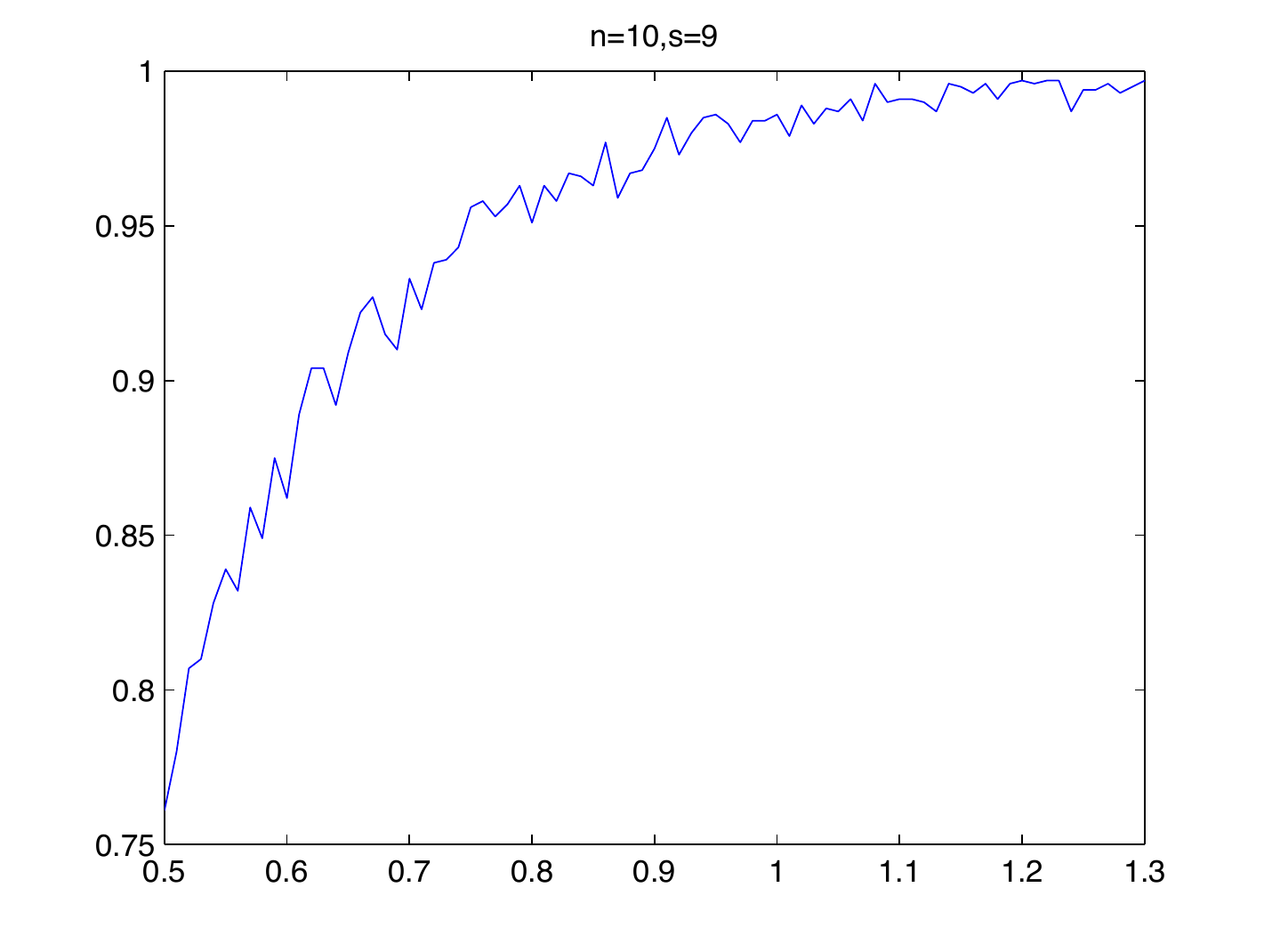}
\includegraphics[width=0.3\textwidth]{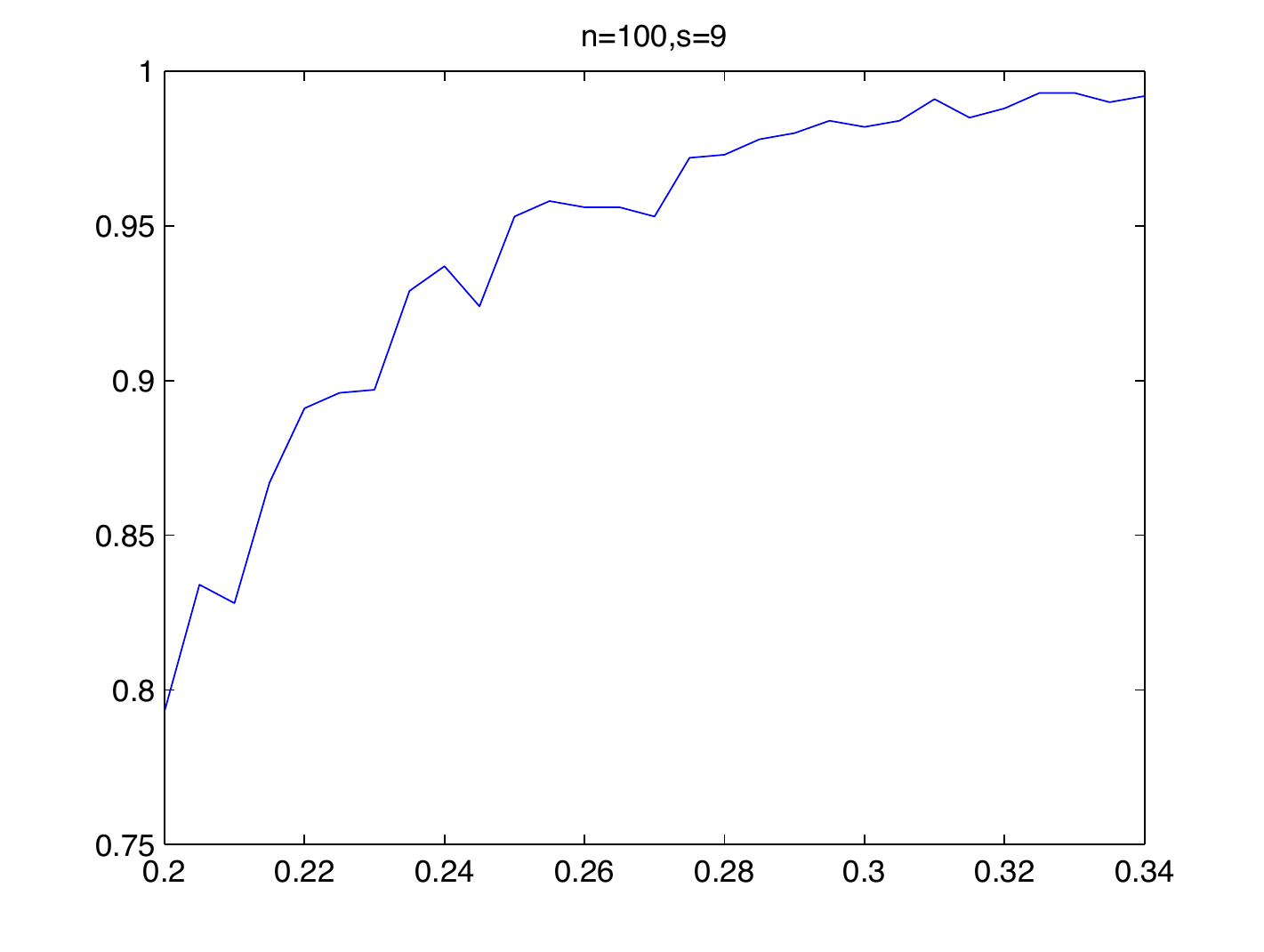}
\includegraphics[width=0.3\textwidth]{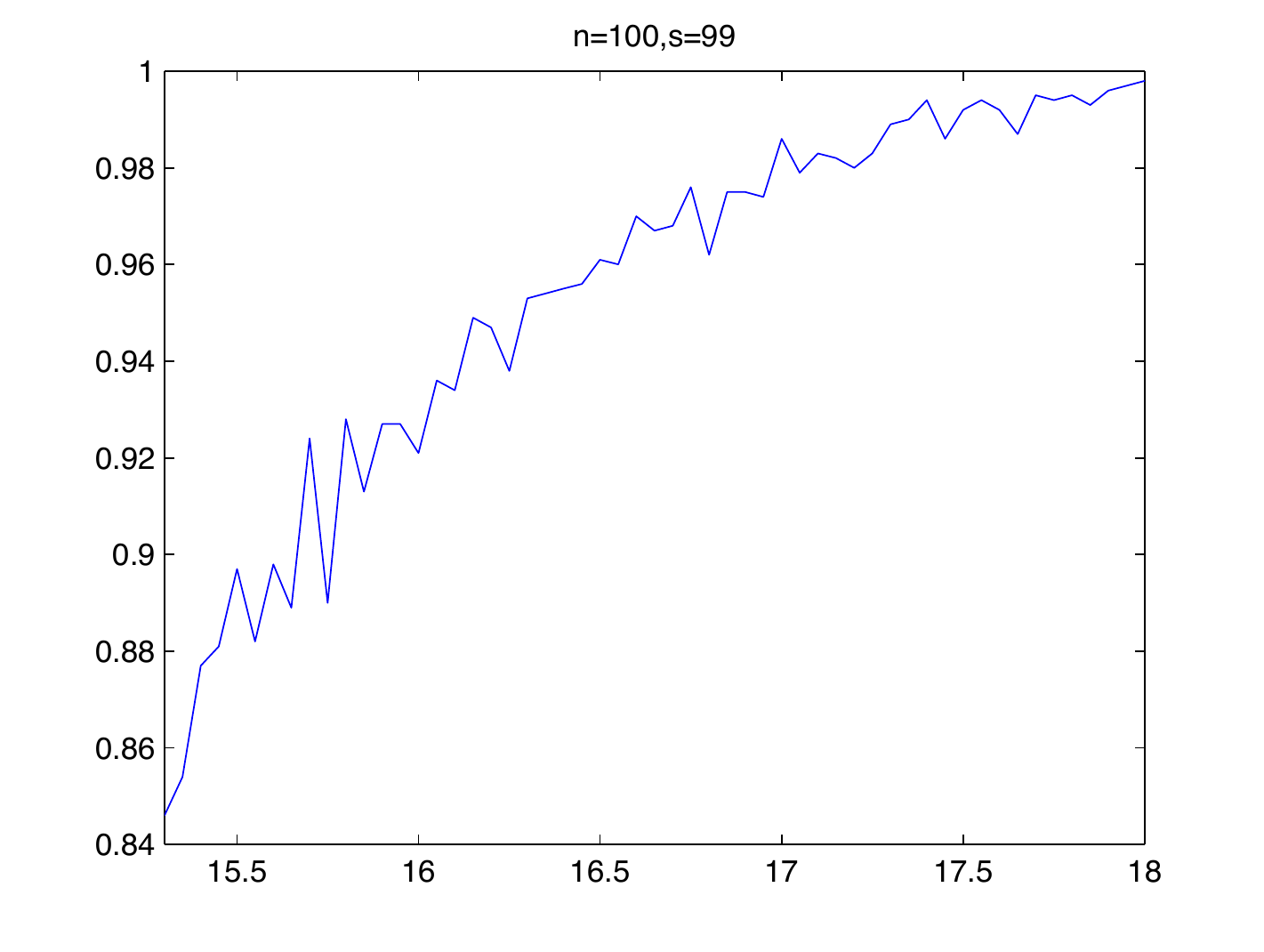}
\end{center}
\caption{Success probability of the MUSIC reconstruction versus aperture for $n=10, s=9$ (left), $n=100, s=9$ (middle)
and $n=100, s=99$ (right).  Note the different  aperture ranges for
the three plots. The success rate is calculated from 1000 trials. Increasing the number of transceivers for the
same number of scatterers reduces the   aperture
required for  the same success rate. The reduction of aperture is
about three folds (left to middle).  On the other hand, higher
number of scatterers with the same number of transceivers also demands larger aperture
for the same success rate. The increase in aperture is about
7 times (middle to right).  }
\label{fig4}
\end{figure}
 \begin{figure}[t]
\begin{center}
\includegraphics[width=0.3\textwidth]{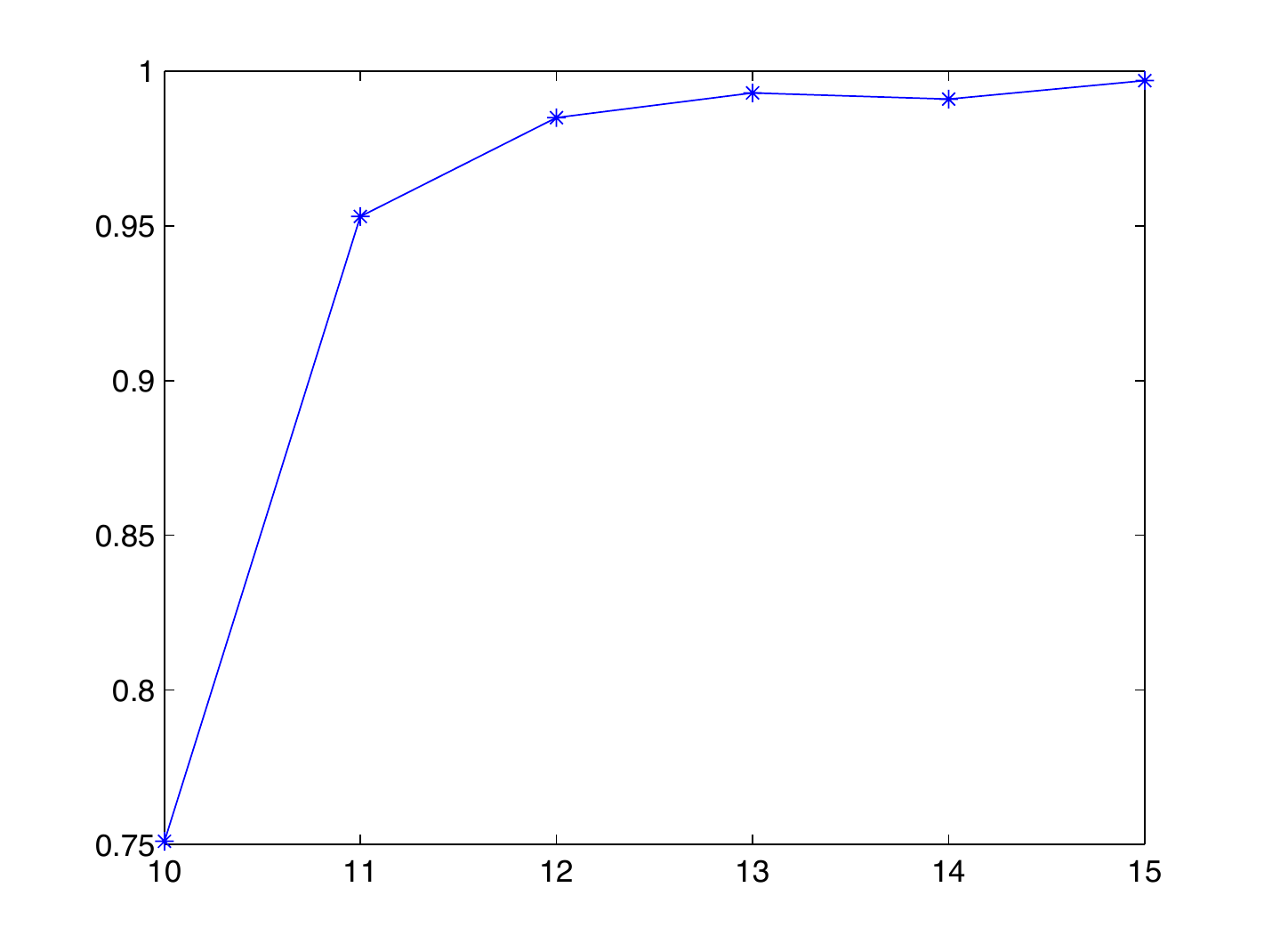}
\includegraphics[width=0.3\textwidth]{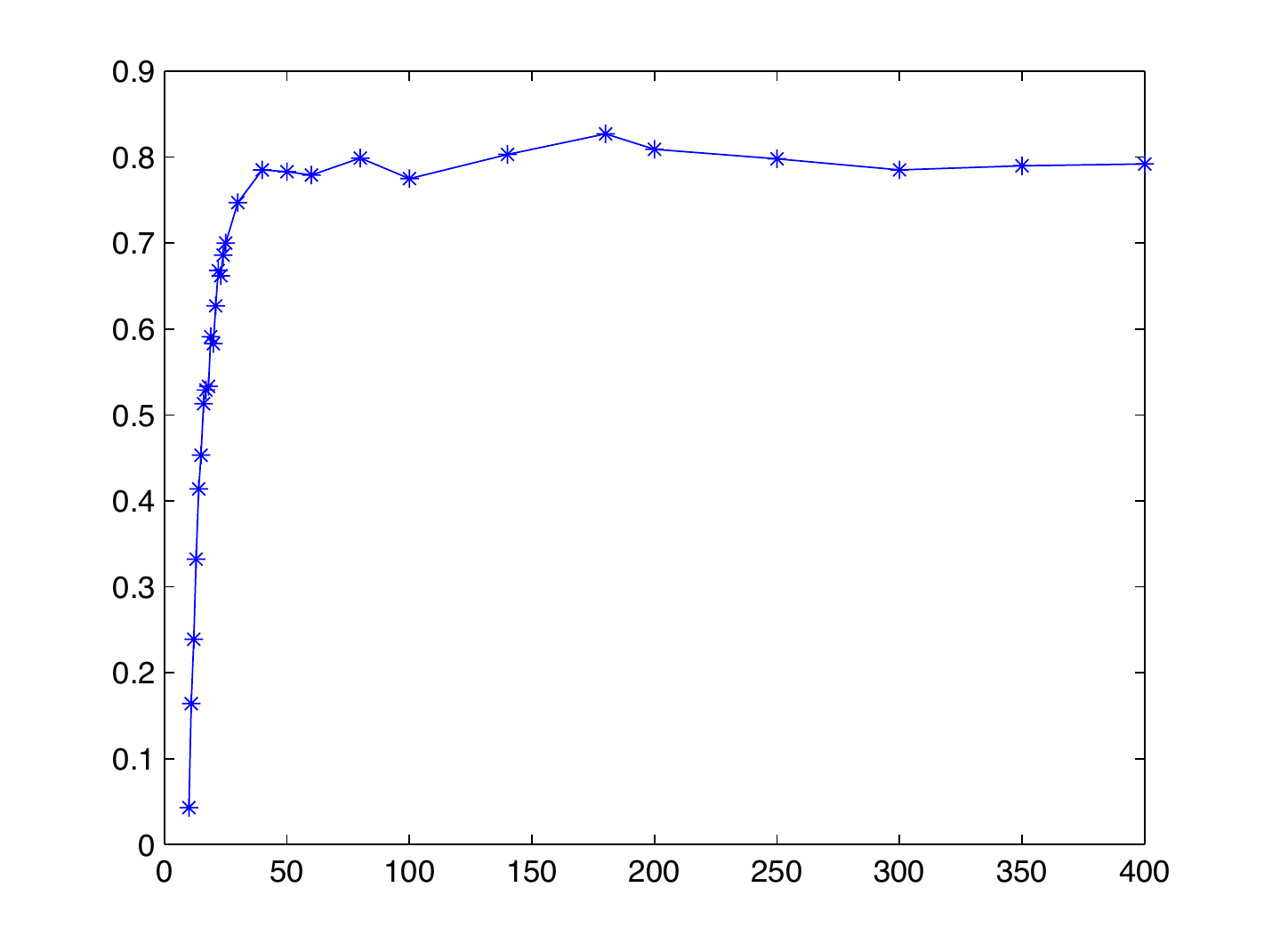}
\includegraphics[width=0.3\textwidth]{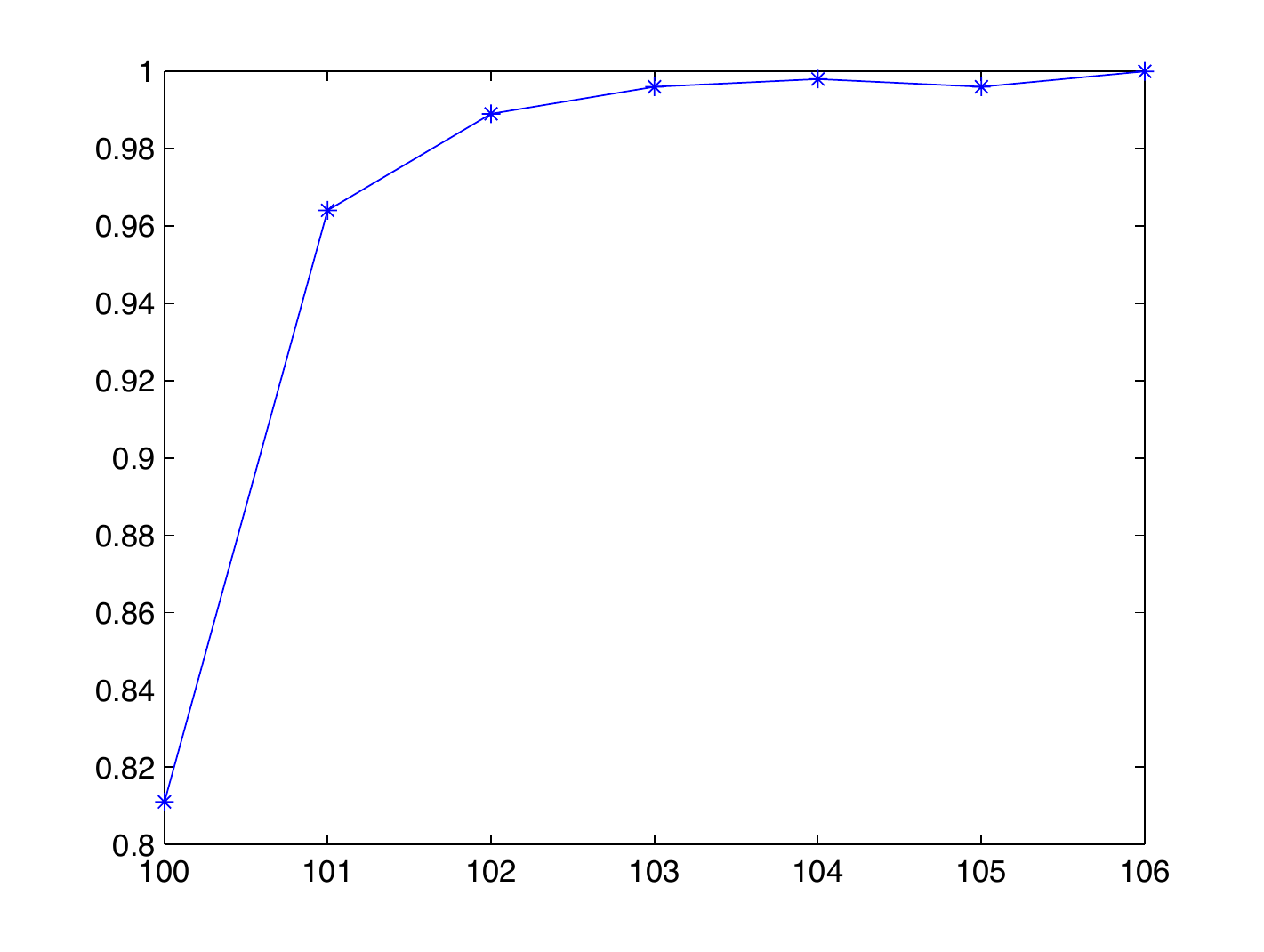}
\end{center}
\caption{Success probability of MUSIC versus the number
of transceivers with $A=0.5, s=9$ (left), 
$A=0.2, s=9$ (middle) and $A=15, s=99$ (right). The probabilities are calculated  from
1000 independent trials.}
\label{fig6}
\end{figure}

We further investigate the performance of  the MUSIC algorithm for the extremely under-resolved case when BP essentially has extremely low probabilities of exact 
recovery (even  for $s=1$). Figure \ref{fig4}
shows the success probabilities of MUSIC as a function
of aperture for various $n$ and $s$ while Figure \ref{fig6} shows the success probabilities  of MUSIC as
a function of $n$ for various $A$ and $s$.  
 The success rates are
calculated from 1000 independent realizations of
transceivers and scatterers.

 Three 
observations about Figure \ref{fig4} are in order: (i) The optimal performance of $s=n-1$ does not hold with certainty  
for  $s$ relatively large with respect to aperture (cf. left  and right panels);  (ii) 
Increasing the number of randomly selected
transceivers reduces the aperture required for the same probability of recovering the same number of scatterers (left to middle panels);
(iii) Increasing the number of randomly selected scatterers
increases the aperture required for
the same  probability  of recovery with the same
number of transceivers (middle to right panels). 

\commentout{
In light of the discussion in Section \ref{sec:grid} regarding
resolution, the fluctuating nature of the success probability   shown in Figure \ref{fig4}  may be caused by insufficiently separated scatterers which are
more likely to cluster as their number increases. The
minimum inter-scatterer   separation necessary for  exact recovery  is determined by
 the frequency and  the aperture (or the sampling distribution $f^{\rm s}$). The threshold separation increases as
 the frequency and the aperture decrease. 
 }

Likewise, the success rates increase with
the number of transceivers for any aperture
and sparsity (Figure \ref{fig6}). The most
interesting plot in Figure \ref{fig6} is the
middle panel which shows for $A=0.2, s=9$
the success rate curve becomes a plateau
after reaching $80\%$. This is {\em not}
inconsistent with the prediction of Proposition \ref{prop1}
since Proposition \ref{prop1} assumes
a {\em fixed} configuration of scatterers while
Figure \ref{fig6} is for random,  independent
realizations of scatterers. In other words
the threshold $n_0$ in Proposition \ref{prop1}  may {\em not} 
be uniformly valid
for all configurations of $s$ scatterers in the under-resolved case. On the other hand, when the aperture  increases by two and half times to $A=0.5$ and the number of transceivers  increases  to 15, the performance  becomes
uniform with respect to the scatterer configuration (left panel).

\begin{figure}[t]
\begin{center}
\includegraphics[width=0.45\textwidth]{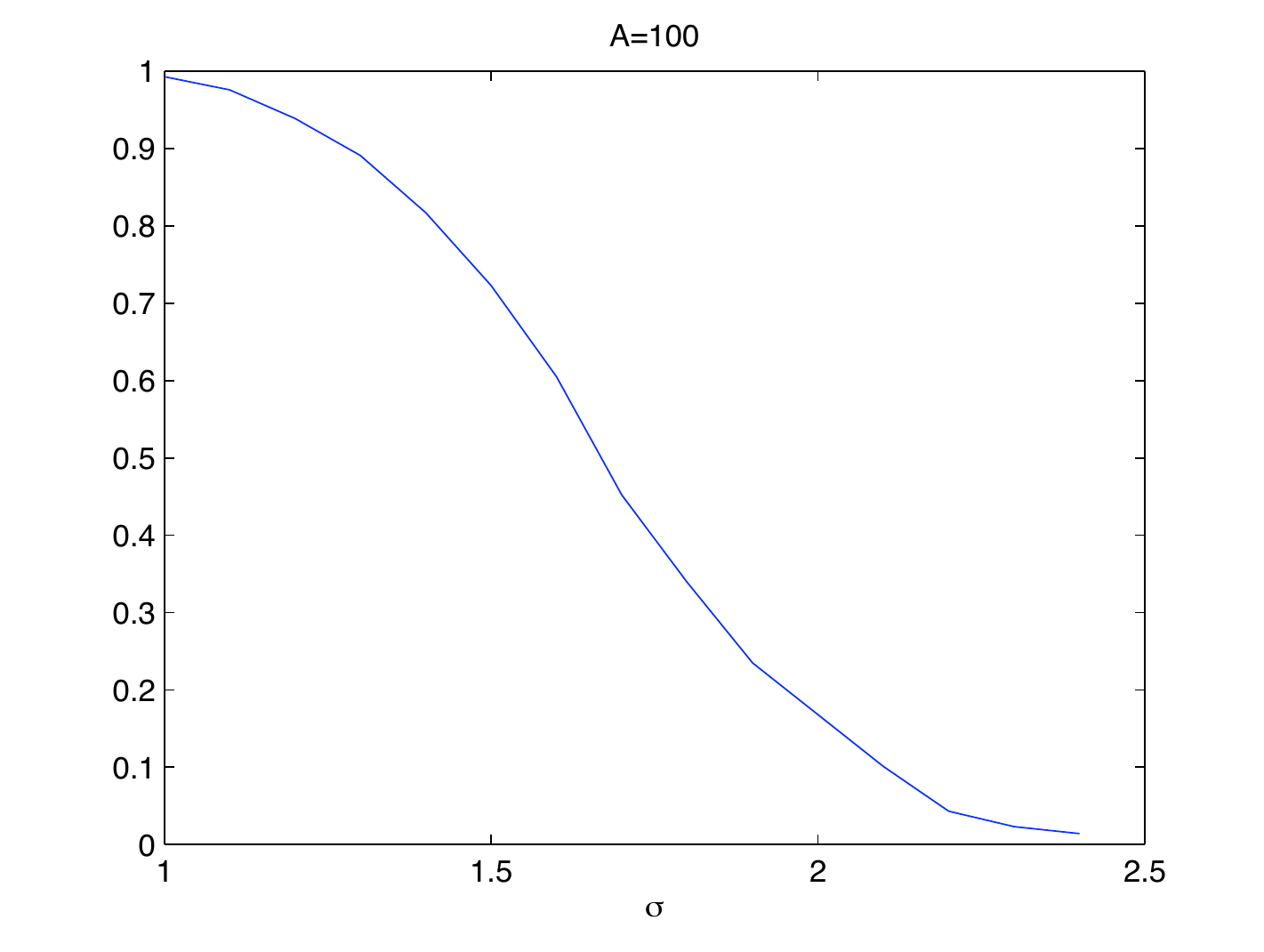}
\includegraphics[width=0.45\textwidth]{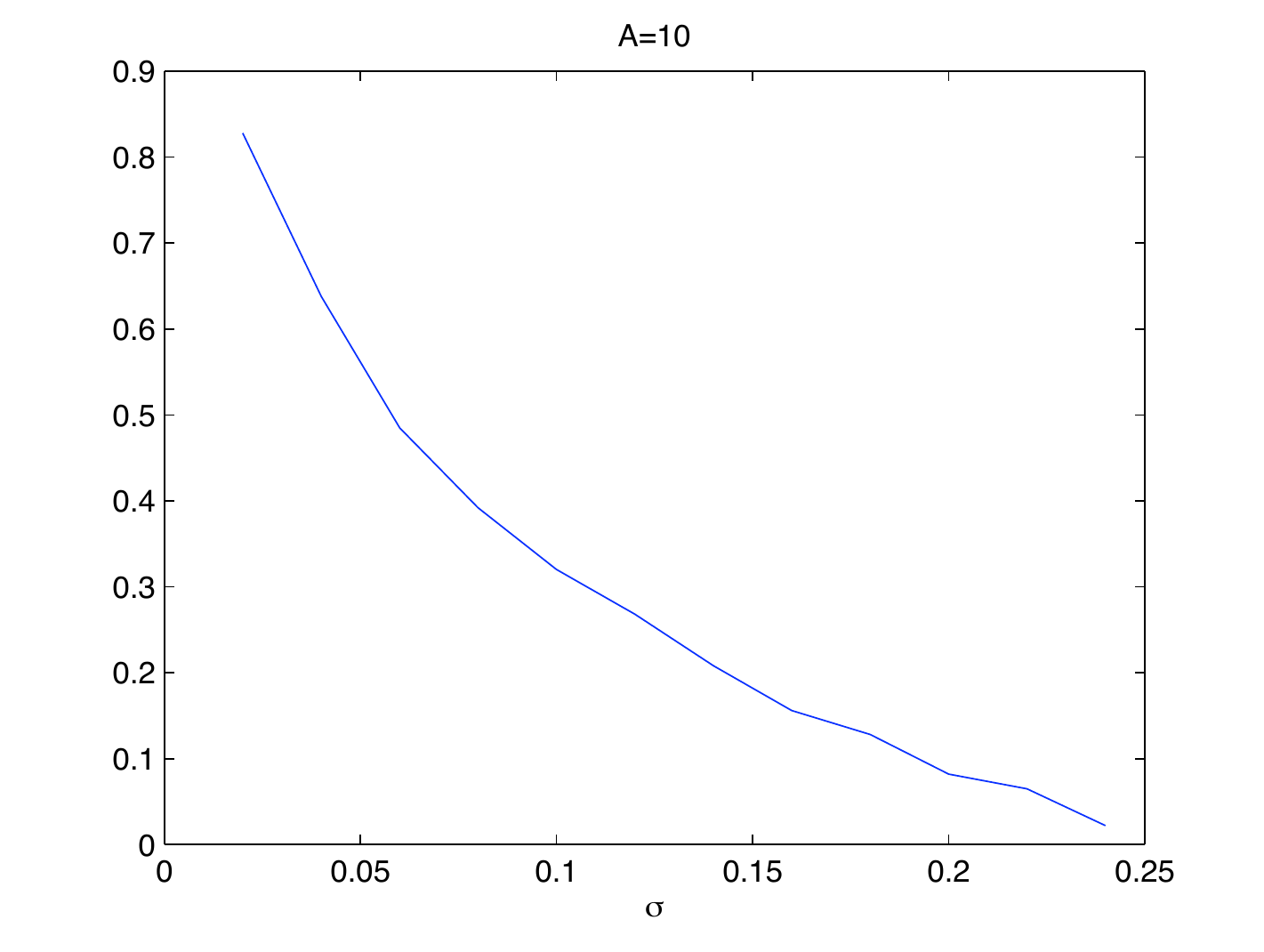}
\end{center}
\caption{Success probability of MUSIC reconstruction of $s=10$ scatterers with  $n=100$ transceivers  versus the noise level $\sigma$ in the well-resolved case $A=100$ (left) and  the
under-resolved case $A=10$ (right). The success rate is calculated
from 1000 trials.  Note the different scales of $\sigma$ in the two plots. Noise sensitivity 
increases dramatically in the under-resolved case.  }
\label{fig5}
\end{figure}

Figure \ref{fig5} shows the noise sensitivity  of MUSIC reconstruction of
10 scatterers 
with 100 transceivers. Here $n$ and $s$ are chosen so that
(\ref{rau2}) is roughly  satisfied.  We add the i.i.d. noises 
\beq
\label{301}
\sigma (e_1+ i e_2) Y_{\rm max}
\eeq
to the entries of the unperturbed data matrix
where $ e_1$ and $e_2$ are independent,  uniform r.v.s in $[-1,1]$ and
$Y_{\rm max}$  is the maximum absolute value of 
the data entries. Hence the signal-to-noise ratio (SNR) is
about $2^{-1}\sigma^{-2}$. In the well-resolved case ($A=100$) the MUSIC reconstruction can withstand 
a significant amount of noise in the data matrix. Indeed,
 at SNR $0.5$
the success rate is almost $100\%$, consistent
with the prediction of Theorem \ref{cor2.1}, and even 
at SNR $0.22$ ($\sigma=1.5$) the success rate can
be indefinitely improved by increasing the number of
transceivers (Figure \ref{fig7}, left panel).  

In the under-resolved case, however, the noise sensitivity
increases significantly. Figure \ref{fig5} (right panel) reminds us how
fragile the superior performance of MUSIC in the under-resolved case is, cf. Figure \ref{fig1} (right panel). 
Figure \ref{fig7} (right panel)  further indicates that 
in the under-resolved case the success rate may not
be  indefinitely improved by increasing the number of
transceivers in the presence of noise.

\begin{figure}[t]
\begin{center}
\includegraphics[width=0.45\textwidth]{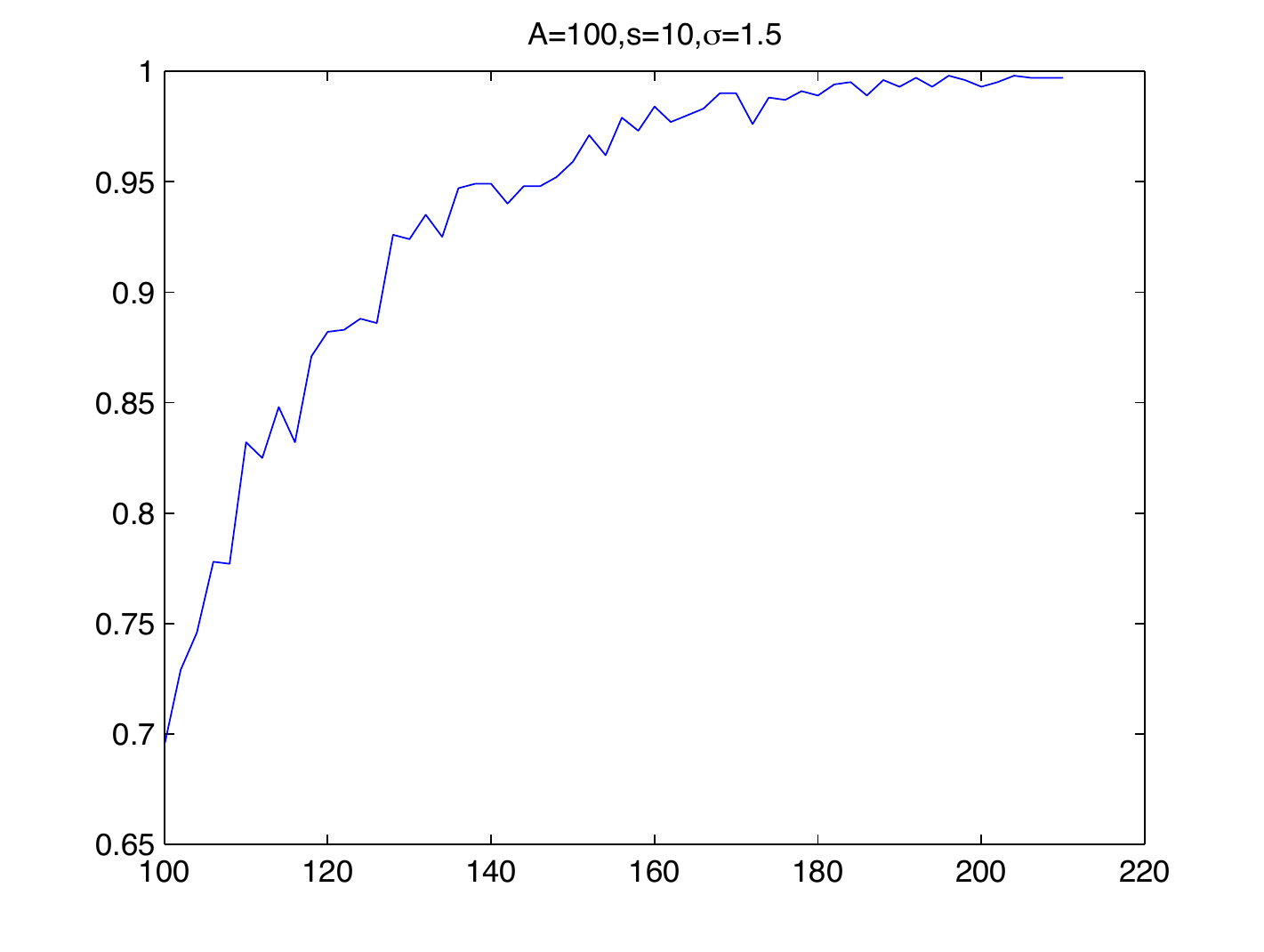}
\includegraphics[width=0.45\textwidth]{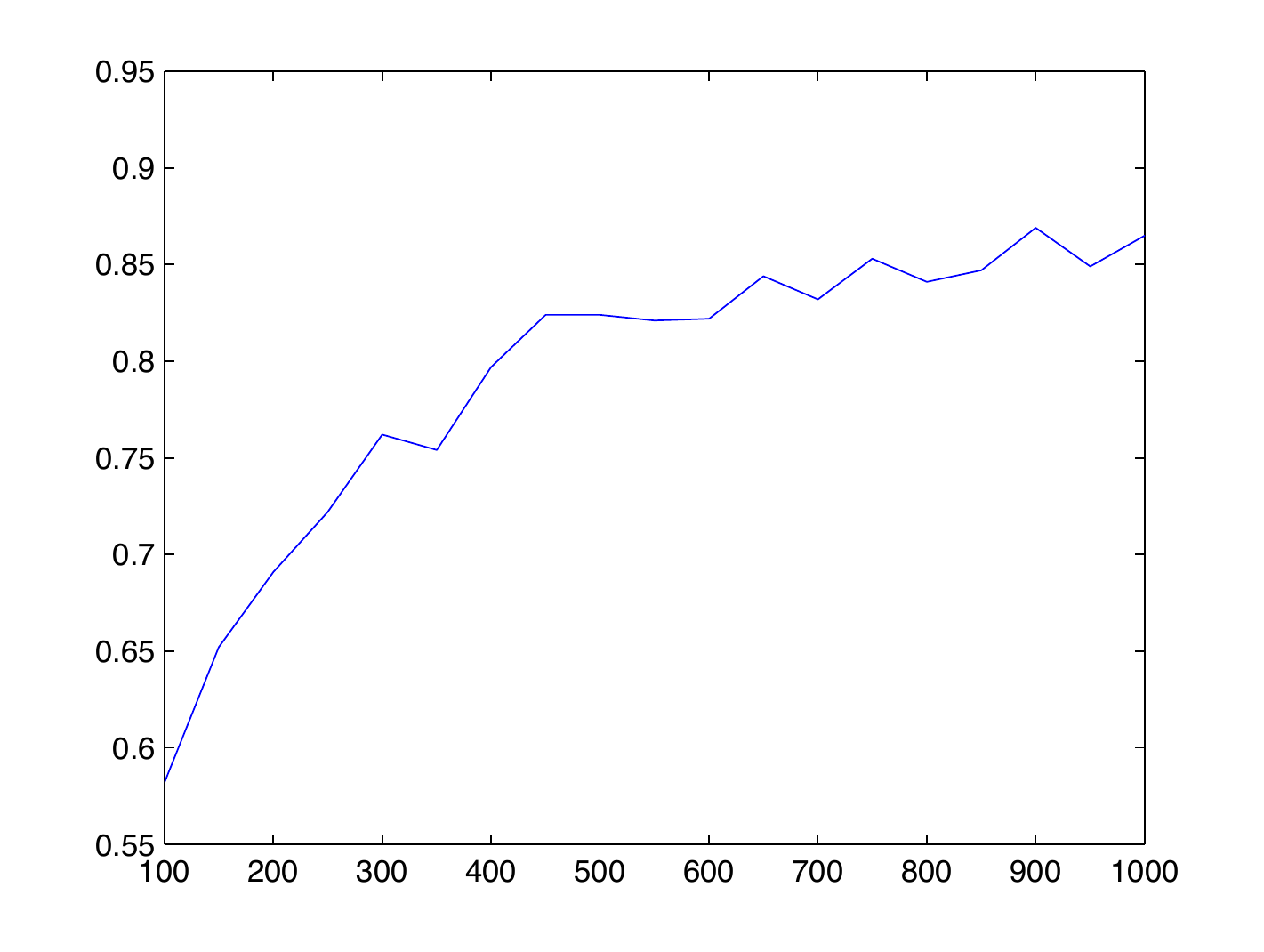}
\end{center}
\caption{Success probability of MUSIC reconstruction of $s=10$ scatterers as a function of  $n$ with $\sigma=150\%$ in the well-resolved case $A=100$ (left) and  $\sigma=5\% $ in the
under-resolved case $A=10$ (right). The success rate reaches 
the plateau of $85\%$ near  $n=1000$ in the under-resolved case.  The success rate is calculated
from 1000 trials. }
\label{fig7}
\end{figure}

\commentout{
Finally, we demonstrate the effect of
dynamic range of scatterers on the
success rate of exact recovery, cf. (\ref{300}).
For the purpose of comparison, 
we set
the noise to be
\beq
\label{301'}
0.5 e^{-5}\sigma (e_1+ie_2)
\eeq
for the dynamic ranges  $[1,2]$ and
$[1,20]$ instead of noise of the form (\ref{301}). 
  Figure \ref{fig8} shows clearly
  that as the dynamic ranges increases the success rate
  decreases accordingly. 

\begin{figure}[t]
\begin{center}
\includegraphics[width=0.45\textwidth]{dynamic-range.pdf}
\includegraphics[width=0.45\textwidth]{dyn-ran.pdf}
\end{center}
\caption{Success probability of MUSIC reconstruction of $s=10$ scatterers with  $n=100$ transceivers  versus the noise level $\sigma$ in the well-resolved case for 
the  range $[1,2]$ (solid curve) and the range
$[1,20]$ (dashed curve) of scatterer amplitudes. 
The left panel shows the result for the noise model (\ref{301})
and the right panel shows the result for the noise model
(\ref{301'}). 
The success rate decreases  
 when the dynamic range of scatterers increases.  
 }
\label{fig8}
\end{figure}
}

\commentout{
\begin{figure}[t]
\begin{center}
\includegraphics[width=0.45\textwidth]{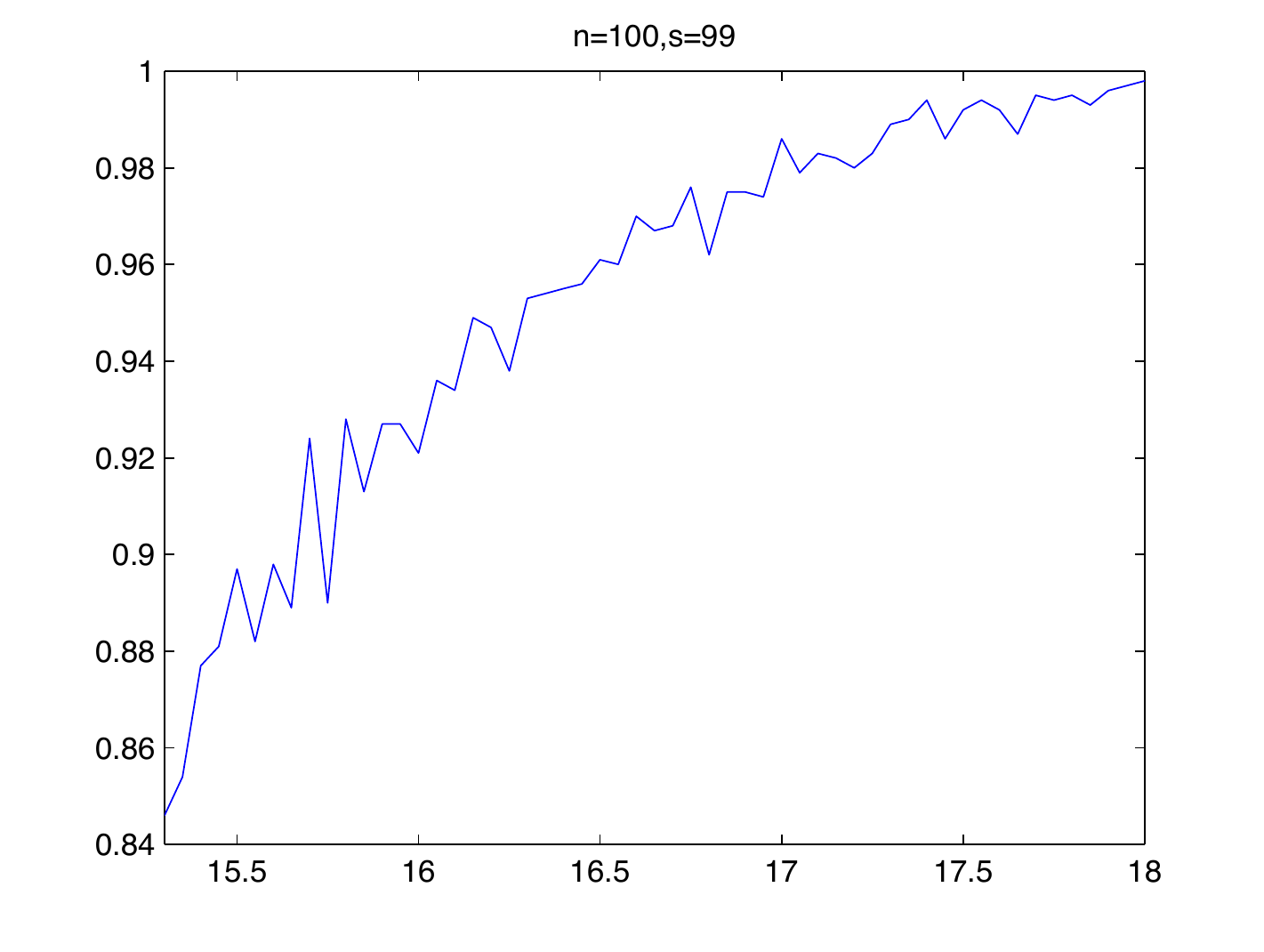}
\end{center}
\caption{Success probability of the MUSIC reconstruction versus aperture: In comparison to Figure \ref{fig4} (left), the recovery of $s=n-1$ requires a larger aperture for larger $s$. The case of $s=99=n-1$ requires
more than 20 times the aperture for the case of $s=9=n-1$.  }
\label{fig5}
\end{figure}
}

\commentout{
\section{The case of rank deficiency}
In addition to the flexibility of grid spacing, yet another
advantage of MUSIC is its low computational complexity. The
MUSIC algorithm amounts to essentially the singular value decomposition
of an $n\times s$ (data) matrix, yet MUSIC requires that
$\bPsi$ has full rank. 
}

\section{Conclusion}\label{sec:con}
We have developed a framework for discrete,  quantitative analysis
of the MUSIC algorithm  in the well-resolved case. 
Our approach  is based on the RIP 
(\ref{rip}) 
and its variant (\ref{rip'}) which takes into account of
the object configuration as well as sparsity. 

Our first main result is a support recovery condition (Theorem \ref{cor2.1}) that for  the {\rm NOR} obeying  (\ref{87})
MUSIC can exactly localize the objects with  noisy data.
Our result indicates the superresolution capability 
of the MUSIC algorithm when the noise level is sufficiently
low (Remark \ref{rmk72}).

We have provided a coherence approach to estimating RIC (Propositions \ref{thm1} and \ref{prop2}) 
for  general  object configuration in three dimensions
with the grid spacing 
$\ell\sim \lambda s$ and the sensor number $ n\sim {s^2}$. When the scatterers
are distributed in a transverse plane, then $\ell\sim \lambda, n\sim s$ (modulo logarithmic factors), suffices.
We have extended these results to the gridless setting
for which $\ell$ is interpreted as the minimum distance  between objects and only approximate localization up to the error $\ell$
is sought (Theorem \ref{thm4} and Corollary \ref{cor10}). 

Our comparative analysis shows that when the whole data matrix is employed in both  BP and
MUSIC, BP outperforms MUSIC in the well-resolved case
in the sense  that
the number of objects  recoverable by BP grows quadratically 
with the number of transceivers while
that by MUSIC grows linearly. The MUSIC reconstruction 
can tolerate a significant amount of noise in the data matrix
(Figures \ref{fig5},  \ref{fig7}, left panels). On the other hand, 
our numerical results show that in the under-resolved case
MUSIC outperforms BP by a wide margin (Figure \ref{fig2}, right panel, Figures \ref{fig4} and \ref{fig6}). However,
 MUSIC's superresolution effect is still unstable with respect to noise in    the data matrix (Figures \ref{fig5}, \ref{fig7},  right panels). 
 
 Finally,  even in the well-resolved case where the employment of just
 {\em one } column of the data matrix by BP guarantees 
 a probabilistic  recovery of objects numbered in  linear
 proportion to the number of sensors, analogous to
 the performance guarantee of MUSIC, 
 the latter outperforms the former in numerical simulations by a wide margin
 (Figure \ref{fig3}).

  \bigskip
 {\bf Acknowledgement.} 
 I am grateful to Mike Yan for preparing the figures
 and  the National Science Foundation  for supporting 
 the research through grant DMS 0908535. 
 I thank Wenjing, Liao for pointing out the result,
 Proposition \ref{prop2'}, which helps improve
 the results of the previous version of the manuscript.

 \appendix
  \section{Sparse extended objects}\label{sec:ext-loc}
 
  \begin{figure}[t]
\begin{center}
\includegraphics[width=0.45\textwidth]{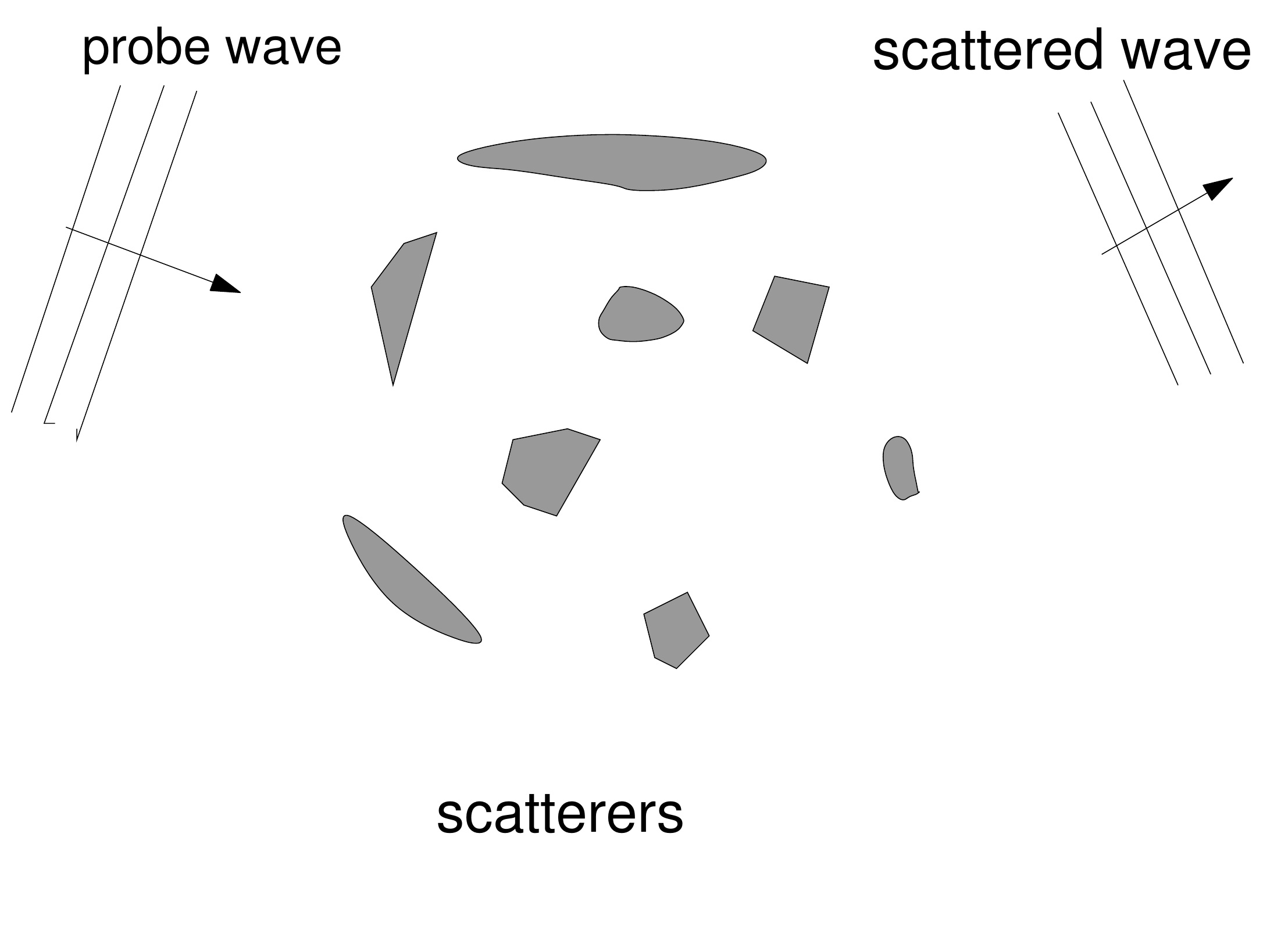}
\end{center}
\caption{Scattering by extended objects}
\label{fig1}
\end{figure}

In this appendix  we extend the MUSIC algorithm to
image sparse extended scatterers by
 interpolating from grid points.

Suppose that the object function $\xi(\br)$ 
 has  a compact support. 
Consider the discrete approximation by interpolating
from the grid points
\beqn
\label{533}
\xi_{\ell}(\br)=\ell^2\sum_{\bq\in \cI}
g(\br/\ell-\bq ) \xi(\ell\bq),\quad \cI\subset \IZ^d
\eeqn
where $g$ is some spline function and $\ell$ is
the grid spacing. Since $\xi$ has a  compact support, $\cI$ is a finite set. 
For simplicity assume $d=2$ and let  $\cI$ be the finite lattice
\[
\cI=\{\bq=(q_1, q_2): q_1, q_2= 1,...,\sqrt{N}\}
\]
of total cardinality $N$ and $\cK=\ell \cI$. 
 In the case of a characteristic function $g$,
$\xi_l$ is a piece-wise constant object function. We will neglect the discretization error and assume $\xi_\ell(\br)=\xi(\br)$ in the subsequent analysis.

The data matrix $\bY\in \IC^{n\times m}$ is given by
\beq
\label{12} Y_{k,l}&\sim &{\ell^d} \sum_{\bq\in I} \xi(\ell\bq)
\int_{\IR^d} g(\br'/\ell-\bq) e^{i\om(\bd_l-\hat\bs_k)\cdot \br'} d\br',\\
&=&{\ell^d(2\pi)^{d/2}} \hat g(\ell\om(\bd_l-\hat\bs_k)) \sum_{\bq\in I}
\xi(\ell\bq)
 e^{i\om\ell(\bd_l-\hat\bs_k)\cdot \bq}.\nn
\eeq
As before  we maintain the option of normalizing $\bY$. 
Suppose $\{\xi_j=\xi(\ell\bq_j): j=1,...,s\} $ is the set
of nonvanishing $\xi (\ell \bq)$ and let $\bX=\hbox{\rm diag}(\xi_j)\in \IC^{s\times s}$. 
Dividing (\ref{12})  by ${\ell^d(2\pi)^{d/2}} \hat g(\ell\om(\bd_l-\hat\bs_k))/2 $ we can write the data matrix 
in the form (\ref{31}) with  the sensing  matrices
\beqn
\label{13'}
\Phi_{kj}
&=& {1\over \sqrt{n}}e^{-i\om \ell\hat\bs_k\cdot\bq} \in \IC^{n\times s}\\
\Psi_{lj}&=& {1\over \sqrt{n}}e^{-i\om\ell \bd_l\cdot \bq}\in \IC^{n\times s} \label{14'}
\eeqn
where $j=(q_1-1)\sqrt{N}+q_2$.

\commentout{
\section{Extension by  frames}
For certain types of objects, the  representation  in terms of
certain  frames  may result in sparse coefficients.

Let 
\[
g_{\bp,\bq}(\br)=e^{if\bp\cdot \br} g(\br-\bq\ell)
\]
and assume that $g_{\bp,\bq}$ constitute a frame. Then
any function $\xi\in L^2$ can be written as
\[
\xi(\br)=\sum_{\bp, \bq\in \IZ^2}
\xi_{\bp,\bq} g_{\bp,\bq}.
\]
The most well-known example is the Gabor frame
generated by
\[
g(\br)=\pi^{-1/2} e^{-|\br|^2/2}
\]
with $f\ell\leq 2\pi$ \cite{JMR}. 
The data matrix $\bY$ has the $(k,l)$-entry
\beq
\label{15} Y_{k,l}&=&2\pi \sum_{\bp,\bq\in \IZ^2}\xi_{\bp,\bq} \hat g(\om(\bd_l-\hat\bs_k) -f\bp)  e^{-i\om \ell(\bd_l-\hat\bs_k)\cdot \bq}. \nn
\eeq
}

In other words, the scattering  analysis for both
point and extended scatterers 
leads to  the same  type of Fourier-like matrices. 

\section{Proof of Theorem \ref{prop4}}\label{app:B}
The following lemma differs  from 
the original version in \cite{Can}. 
\begin{lemma} We have
\[
\lt|\Re\lan \tilde\bA Z,\tilde\bA Z'\ran\rt|
\leq {1\over 2}\lt(\delta^+_{s+s'} +\delta_{s+s'}^-\rt)\|Z\|_2\|Z'\|_2
\]
for all $Z, Z'$ supported on disjoint subsets $T,T'\subset \{1,...,m\} $ with $|S|\leq s, |S'|\leq s'.$
\end{lemma}
\begin{proof} Without loss of generality, suppose $Z,Z'$ are unit vectors.  Since $Z\perp Z'$,  $\|Z\pm Z'\|^2_2=2.$ 
Hence we have from the RIP  (\ref{rip})
\beq
2(1-\delta^-_{s+s'})\leq
\|\tilde\bA(Z\pm Z')\|_2^2\leq 2(1+\delta^+_{s+s'})\label{a.1}
\eeq
By the parallelogram identity and (\ref{a.1}) 
\[
\lt|\Re\lan \tilde\bA Z, \tilde\bA Z'\ran\rt|=
{1\over 4}\lt|\|\tilde\bA Z+\tilde\bA Z'\|_2^2-\|\tilde\bA Z-\tilde\bA Z'\|_2^2\rt|
\leq {1\over 2}\lt(\delta^+_{s+s'} +\delta_{s+s'}^-\rt) \]
which proves the lemma. 
\label{lemma1}
\end{proof}

By the triangle inequality and the fact that $Z$ is in the  feasible
set we have
\beq
\label{a7}
\|\tilde\bA(\hat Z-Z)\|_2\leq \|\tilde\bA \hat Z-Y\|_2+\|Y-\tilde\bA Z\|_2\leq
2\ep.
\eeq
Set $\hat Z=Z+\Delta$ and decompose $\Delta$ into a sum of vectors
$\Delta_{S_0}, \Delta_{S_1}, \Delta_{S_2},...,$ each of sparsity at most $s$.
Here $S_0$ corresponds to the locations of the $s$ largest coefficients of $Z$; $S_1$ to the locations of the $s$ largest
coefficients of $\Delta_{S_0^c}$;
$S_2$ to the locations of the next $s$ largest coefficients of
$\Delta_{S_0^c}$, and so on. 

{\bf Step (i)}. For $j\geq 2$,
\beqn
\|\Delta_{S_j}\|_2\leq s^{1/2} \|\Delta_{S_j}\|_\infty\leq s^{-1/2}
\|\Delta_{S_{j-1}}\|_2
\eeqn
and hence
\beq
\label{a6}
\sum_{j\geq 2} \|\Delta_{S_j}\|_2\leq s^{-1/2} \sum_{j\geq 1}
\|\Delta_{S_j}\|_1\leq s^{-1/2} \|\Delta_{S_0^c}\|_1.
\eeq
This yields
\beq
\label{a.2}
\|\Delta_{(S_0\cup S_1)^c}\|_2=\|\sum_{j\geq 2} \Delta_{S_j}\|_2
\leq \sum_{j\geq 2} \|\Delta_{S_j}\|_2\leq s^{-1/2} \|\Delta_{S_0^c}\|_1.
\eeq
Also we have
\beqn
\|Z\|_1\geq \|\hat Z\|_1=\|Z_{S_0}+\Delta_{S_0}\|_1
+\|Z_{S_0^c}+\Delta_{S_0^c}\|_1\geq
\|Z_{S_0}\|_1-\|\Delta_{S_0}\|_1-\|Z_{S_0^c}\|_1
+\|\Delta_{S_0^c}\|_1
\eeqn
which  implies
\beq
\label{a.3}
\|\Delta_{S_0^c}\|_1\leq 2\|Z_{S_0^c}\|_1+\|\Delta_{S_0}\|_1.
\eeq
Note that $\|Z_{S_0^c}\|_1=\|Z-Z^{(s)}\|_1$ by definition.
Applying (\ref{a.2}), (\ref{a.3}) and the Cauchy-Schwarz 
inequality gives
\beq
\label{a4}
\|\Delta_{(S_0\cup S_1)^c}\|_2\leq \|\Delta_{S_0}\|_2+2e_0
\eeq
where $e_0\equiv s^{-1/2} \|Z-Z^{(s)}\|_1$.

{\bf Step (ii).} Observe
\beq\label{38}
\|\tilde\bA \Delta_{S_0\cup S_1}\|_2^2
&=&\lan \tilde\bA \Delta_{S_0\cup S_1},\tilde \bA \Delta\ran
-\lan \tilde\bA \Delta_{S_0\cup S_1},\sum_{j\geq 2}
\tilde\bA \Delta_{S_j}\ran\\
&=&\Re \lan\tilde \bA \Delta_{S_0\cup S_1}, \tilde\bA \Delta\ran
-\sum_{j\geq 2}\Re \lan \tilde\bA \Delta_{S_0\cup S_1},
\tilde\bA \Delta_{S_j}\ran\nn\\
&=&\Re \lan \tilde\bA \Delta_{S_0\cup S_1},\tilde \bA \Delta\ran
-\sum_{j\geq 2}\lt[\Re \lan\tilde \bA \Delta_{S_0},
\tilde\bA \Delta_{S_j}\ran+ \Re \lan \tilde\bA \Delta_{S_1},
\tilde \bA \Delta_{S_j}\ran\rt].\nn
\eeq
This calculation  differ slightly  from  the corresponding calculation in \cite{Can}. 

From (\ref{a7}) and the RIP (\ref{rip})  it follows that
\[
\|\lan \tilde\bA \Delta_{S_0\cup S_1}\ran\|\leq \|\tilde\bA \Delta_{S_0\cup S_1}\|_2
\|\tilde \bA \Delta\|_2\leq 2\ep \sqrt{1+\delta^+_{2s}}\|\Delta_{S_0\cup S_1}\|_2.
\]
Moreover, it follows from Lemma \ref{lemma1}
that
\beqn
\lt|\Re\lan\tilde \bA \Delta_{S_0}, \tilde\bA \Delta_{S_j}\ran \rt|
&\leq &{1\over 2}\lt(\delta^+_{2s} +\delta^-_{2s}\rt)\|\Delta_{S_0}\|_2 \|\Delta_{S_j}\|_2\\
\lt|\Re\lan \tilde \bA \Delta_{S_1}, \tilde \bA \Delta_{S_j}\ran \rt|
&\leq &{1\over 2}\lt(\delta^+_{2s} +\delta^-_{2s}\rt) \|\Delta_{S_0}\|_2 \|\Delta_{S_j}\|_2
\eeqn
for $j\geq 2$. Since $S_0$ and $S_1$ are disjoint
\[
\|\Delta_{S_0}\|_2+\|\Delta_{S_1}\|_2\leq \sqrt{2}
\sqrt{\|\Delta_{S_0}\|^2_2+\|\Delta_{S_1}\|^2_2}
=\sqrt{2}  \|\Delta_{S_0\cup S_1}\|_2.
\]
Also 
\[
(1-\delta^-_{2s}) \|\Delta_{S_0\cup S_1}\|_2^2
\leq \|\tilde \bA \Delta_{S_0\cup S_1}\|_2^2\leq
\|\Delta_{S_0\cup S_1}\|_2\lt(2\ep \sqrt{1+\delta^+_{2s}}
+{1\over \sqrt{2}}\lt(\delta^+_{2s} +\delta^-_{2s}\rt) \sum_{j\geq 2}
\|\Delta_{S_j}\|_2\rt). 
\]
Therefore from (\ref{a6}) we obtain
\beqn
\label{a8}
\|\Delta_{S_0\cup S_1}\|_2\leq \alpha \ep +\rho s^{-1/2}
\|\Delta_{S_0^c}\|_1,\quad \alpha={2\sqrt{1+\delta^+_{2s}}\over
1-\delta^-_{2s}},\quad \rho={{1\over \sqrt{2}}\lt(\delta^+_{2s} +\delta^-_{2s}\rt)\over
1-\delta^-_{2s}}
\eeqn
and moreover by (\ref{a.3}) 
\[
\|\Delta_{S_0\cup S_1}\|_2\leq \alpha \ep +\rho 
\|\Delta_{S_0^c}\|_2+2\rho e_0.
\]
Namely,
\[
\|\Delta_{S_0\cup S_1}\|_2\leq (1-\rho)^{-1}(\alpha \ep+2\rho e_0)
\]
if (\ref{ric}) holds. 

Finally,
\beqn
\|\Delta\|_2\leq \|\Delta_{S_0\cup S_1}\|_2+\|\Delta_{(S_0\cup S_1)^c}\|_2
\leq 2\|\Delta_{S_0\cup S_1}\|_2+2 e_0\leq
2(1-\rho)^{-1} (\alpha\ep+ (1+\rho)e_0)
\eeqn
which is what we set out to show.

\commentout{
A useful by-product of the proof is the following. 
\begin{lemma} 
Let $\Delta$ be any null vector of $\bA$ and let $S_0$ be any set of cardinality $s$. Then
\[
\|\Delta_{S_0}\|_1\leq \rho \|\Delta_{S_0^c}\|_1.
\]
\end{lemma}
\begin{proof} Consider $\ep=0$ and $\hat Z=Z+\Delta$. We have
\[
\|\Delta_{S_0}\|_1\leq s^{1/2}\|\Delta_T{S_0}\|_2\leq s^{1/2}
\|\Delta_{S_0\cup S_1}\|_1.
\]
Now the desired result follows from  (\ref{a8}).
\end{proof}
}

\end{document}